\documentclass[12pt]{article}

\usepackage{amssymb}
\usepackage{amsmath}
\usepackage{fullpage}
\usepackage{graphicx}

\newtheorem{theorem}{Theorem}
\newtheorem{lemma}{Lemma}
\newtheorem{remark}{Remark}
\newtheorem{observation}{Observation}

\newcommand{\qed}{\rule{0.5em}{1.5ex}}
\newcommand{\fqed}{{\hfill~\qed}}
\newenvironment{proof}{{\noindent \bf Proof.}}
                      {{\hfill \fqed} \vspace{1em}}

\newcommand{\lca}{\mathsf{LCA}}
\newcommand{\level}{\mathord{\it level}}
\newcommand{\dst}{\mathsf{dist}}
\newcommand{\ps}{\mathsf{PS}}
\newcommand{\dd}{\mathsf{D}}
\newcommand{\bsd}{\mathsf{BD}}
\newcommand{\spp}{\mathsf{SP}}
\newcommand{\cw}{\mathsf{CW}}
\newcommand{\ccw}{\mathsf{CCW}}

\title{Shortest Beer Path Queries in Outerplanar 
Graphs\thanks{A preliminary version was presented at the 32nd Annual 
International Symposium on Algorithms and Computation (ISAAC 2021).
JB was supported by an NSERC Undergraduate Student Research Award. 
MS was suported by NSERC.}} 
\author{
Joyce Bacic\thanks{School of Computer Science, Carleton University, 
Ottawa, Canada.} 
\and
Saeed Mehrabi\thanks{UMass Lowell, USA}
\and
Michiel Smid\footnotemark[2]
}

\begin{document}

\maketitle

\begin{abstract}
A \emph{beer graph} is an undirected graph $G$, in which each edge has
a positive weight and some vertices have a beer store. 
A \emph{beer path} between two vertices $u$ and $v$ in $G$ is any path
in $G$ between $u$ and $v$ that visits at least one beer store.

We show that any outerplanar beer graph $G$ with $n$ vertices can be
preprocessed in $O(n)$ time into a data structure of size $O(n)$, such
that for any two query vertices $u$ and $v$, (i) the weight of the
shortest beer path between $u$ and $v$ can be reported in $O(\alpha(n))$
time (where $\alpha(n)$ is the inverse Ackermann function), and
(ii) the shortest beer path between $u$ and $v$ can be reported in
$O(L)$ time, where $L$ is the number of vertices on this path.
Both results are optimal, even when $G$ is a beer tree
(i.e., a beer graph whose underlying graph is a tree).
\end{abstract}

\section{Introduction} 
Imagine that you are going to visit a friend and, not wanting to show 
up empty handed, you decide to pick up some beer along the way. In this 
paper we determine the fastest way to go from your place to your 
friend's place while stopping at a beer store to buy some drinks.

A \emph{beer graph} is a undirected graph $G=(V,E)$, in which each edge 
$(u,v)$ has a positive weight $\omega(u,v)$ and some of the vertices are 
beer stores. For two vertices $u$ and $v$ of $G$, we define the 
\emph{shortest beer path} from $u$ to $v$ to be the shortest 
(potentially non-simple) path that starts at $u$, ends at $v$, and 
visits at least one beer store. We denote this shortest path by 
$\spp_B(u,v)$. The \emph{beer distance} $\dst_B(u,v)$ between $u$ 
and $v$ is the weight of the path $\spp_B(u,v)$, i.e., the sum of the 
edge weights on $\spp_B(u,v)$. 

Observe that even though the shortest beer path from $u$ to $v$ may be 
a non-simple path, it is always composed of two simple paths: the 
shortest path from $u$ to a beer store and the shortest path from this 
same beer store to $v$. Thus, when looking at the shortest beer path 
problem, we often need to consider the shortest path between vertices. 
We denote the shortest path in $G$ from $u$ to $v$ by $\spp(u,v)$ and 
we use $\dst(u,v)$ to denote the weight of this path. We also say that 
$\dst(u,v)$ is the \emph{distance} between $u$ and $v$ in $G$. 

To the best of our knowledge, the problem of computing shortest beer 
paths has not been considered before. Let $s$ be a fixed source vertex
of $G$. Recall that Dijkstra's algorithm computes $\dst(s,v)$ for all
vertices $v$, by maintaining a ``tentative distance'' $\delta(v)$, 
which is the weight of the shortest path from $s$ to $v$ computed so 
far. If we also maintain a ``tentative beer distance'' $\delta_B(v)$
(which is the weight of the shortest beer path from $s$ to $v$ that 
has been found so far), then a modification of Dijkstra's algorithm 
allows us to compute $\dst_B(s,v)$ for all vertices $v$, in 
$O(|V| \log |V| + |E|)$ total time. 

As far as we know, no non-trivial results are known for beer distance 
queries. In this case, we want to preprocess the beer graph $G$ into 
a data structure, such that, for any two query vertices $u$ and $v$, 
the shortest beer path $\spp_B(u,v)$, or its weight $\dst_B(u,v)$, 
can be reported.

\subsection{Our Results}
\label{secOR}
We present data structures that can answer shortest beer path queries in 
outerplanar beer graphs. Recall that a graph $G$ is \emph{outerplanar}, 
if $G$ can be embedded in the plane, such that all vertices are on the 
outer face, and no two edges cross. 

Our first result is stated in terms of the inverse Ackermann function. 
We use the definition as given in~\cite{timothy}: Let $A_0(i) = i+1$ 
and, for $\ell \geq 0$, $A_{\ell+1}(i) = A_{\ell}^{(i+1)}(i+8)$, where 
$A_{\ell}^{(i+1)}$ is the function $A_{\ell}$ iterated $i+1$ times.
We define $\alpha(m,n)$ to be the smallest value of $\ell$ for which
$A_{\ell}(\lfloor m/n \rfloor) > n$, and we define 
$\alpha(n) = \alpha(n,n)$. 

\begin{theorem}  \label{thmdq}  
Let $G$ be an outerplanar beer graph with $n$ vertices.  
For any integer $m \geq n$, we can preprocess $G$ in $O(m)$ time 
into a data structure of size $O(m)$, such that for any two query 
vertices $u$ and $v$, both $\dst(u,v)$ and $\dst_B(u,v)$ can be 
computed in $O(\alpha(m,n))$ time.   
\end{theorem} 

By taking $m=n$, both the preprocessing time and the space used are 
$O(n)$, and for any two query vertices $u$ and $v$, both $\dst(u,v)$ 
and $\dst_B(u,v)$ can be computed in $O(\alpha(n))$ time. 

As another example, let $\log^* n$ be the number of times the function
$\log$ must be applied, when starting with the value $n$, until the
result is at most~$1$, and let $\log^{**} n$ be the number of times the
function $\log^*$ must be applied, again starting with $n$, until the
result is at most~$1$. Let $m = n \log^{**} n$. Since 
$\alpha(m,n) = O(1)$, we obtain a data structure with space and 
preprocessing time $O(n \log^{**} n)$ that can answer both distance 
and beer distance queries in $O(1)$ time.

As we mentioned before, beer distance queries have not been considered 
for any class of graphs. In fact, the only result on (non-beer) distance 
queries in outerplanar graphs that we are aware of is by 
Djidjev \emph{et al.}~\cite{DPZ91}. They show that an outerplanar graph 
with $n$ vertices can be preprocessed in $O(n \log n)$ time into a 
data structure of size $O(n \log n)$, such that any distance query 
can be answered in $O(\log n)$ time. Our result in Theorem~\ref{thmdq}  
significantly improves their result. 

We also show that the result in Theorem~\ref{thmdq} is optimal for beer 
distance queries, even if $G$ is a \emph{beer tree} 
(i.e., a beer graph whose underlying graph is a tree).
We do not know if the query time is optimal for (non-beer) distance 
queries.  

Our second result is on reporting the shortest beer path between two 
query vertices. 

\begin{theorem}  \label{thm2}  
Let $G$ be an outerplanar beer graph with $n$ vertices. 
We can preprocess $G$ in $O(n)$ time into a data structure of size 
$O(n)$, such that for any two vertices $u$ and $v$, the shortest beer 
path from $u$ to $v$ can be reported in $O(L)$ time, where $L$ is 
the number of vertices on this beer path. 
\end{theorem} 

Observe that the query time in Theorem~\ref{thm2} does not depend on 
the number $n$ of vertices of the graph. Again, we are not aware of 
any previous work on reporting shortest beer paths. 
Djidjev \emph{et al.}~\cite{DPZ91} show that, after $O(n \log n)$ 
preprocessing and using $O(n \log n)$ space, the shortest (non-beer) 
path between two query vertices can be reported in $O(\log n + L)$ time, 
where $L$ is the number of vertices on the path. 

\subsection{Preliminaries and Organization} 
Throughout this paper, we only consider outerplanar beer graphs $G$. 
The number of vertices of $G$ is denoted by $n$. It is well known that 
$G$ has at most $2n-3$ edges.  
As in~\cite{DPZ91}, we say that $G$ satisfies the 
\emph{generalized triangle inequality}, if for every edge $(u,v)$ in 
$G$, $\dst(u,v) = \omega(u,v)$, i.e., the shortest path between $u$ and 
$v$ is the edge $(u,v)$. 

The outerplanar graph $G$ is called \emph{maximal}, if adding an edge 
between any two non-adjacent vertices of $G$ results in a graph that is 
not outerplanar. In this case, the number of edges is equal to $2n-3$. 
A maximal outerplanar graph $G$ is 2-connected, each 
internal face of $G$ is a triangle and the outer face of $G$ forms a 
Hamiltonian cycle. In such a graph, edges on the outer face will be 
referred to as \emph{external} edges, where all other edges will be 
referred to as \emph{internal} edges. 

The \emph{weak dual} of a maximal outerplanar graph $G$ is the graph 
$D(G)$ whose node set is the set of all internal faces of $G$, and in 
which $(F,F')$ is an edge if and only if the faces $F$ and $F'$ share 
an edge in $G$; see Figure~\ref{fig:dual}. 
For simplicity, we will refer to $D(G)$ as the dual of $G$. Observe that 
$D(G)$ is a tree with $n-2$ nodes, each of which has degree at most 
three. 

\begin{figure}[h]
  \centering
  \includegraphics[scale = 0.5]{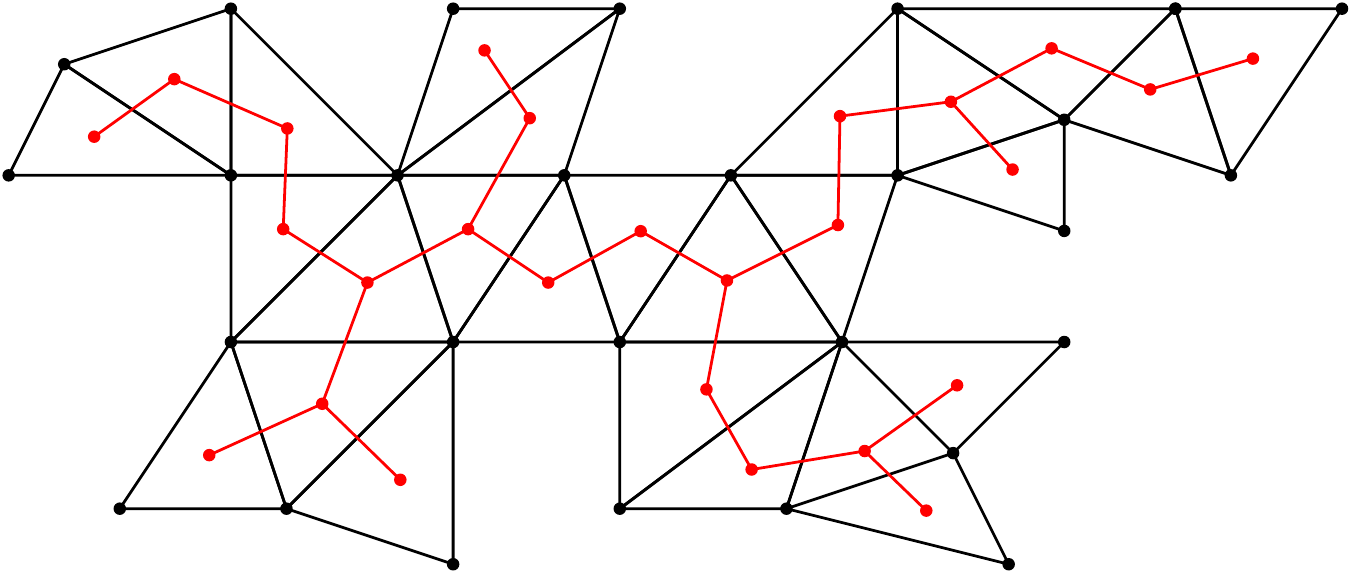}
  \caption{A maximal outerplanar graph shown in black. Its dual is 
      shown in red.}
  \label{fig:dual}
\end{figure}

If $H$ is a subgraph of the beer graph $G$, and $u$ and $v$ are vertices 
of $H$, then $\dst(u,v,H)$ and $\dst_B(u,v,H)$ denote the distance and 
beer distance between $u$ and $v$ in $H$, respectively. The shortest 
beer path in $H$ between $u$ and $v$ must be entirely within $H$. 
Observe that we use the shorthand $\dst(u,v)$ for $\dst(u,v,G)$, 
and $\dst_B(u,v)$ for $\dst_B(u,v,G)$. 

It will not be surprising that the algorithms for computing shortest 
beer paths use the dual $D(G)$. Thus, our algorithms will need some 
basic data structures on trees. These data structures will be 
presented in Section~\ref{sectree}. 

In Section~\ref{secmaxdist}, we will prove Theorem~\ref{thmdq} for 
maximal outerplanar beer graphs. We also prove that the 
result in Theorem~\ref{thmdq} is optimal, even for beer trees.
The proof of Theorem~\ref{thm2}, again for maximal outerplanar beer 
graphs, will be presented in Section~\ref{secJoyce}. 
Both Sections~\ref{secmaxdist} and~\ref{secJoyce} will use the result 
in Lemma~\ref{lem-crazy}, whose detailed proof will be given in 
Section~\ref{app:lem-crazy}. 

The extensions of Theorems~\ref{thmdq} and~\ref{thm2} to arbitrary 
outerplanar beer graphs will be given in Section~\ref{app:extension}. 
Finally, Section~\ref{secSSSBP} will present an $O(n)$-time 
algorithm for computing the single-source shortest beer path tree for 
any given source vertex.

\section{Query Problems on Trees} \label{sectree} 
Our algorithms for computing beer shortest paths in an outerplanar 
graph $G$ will use the dual of $G$, which is a tree. In order 
to obtain fast implementations of these algorithms, we need to be 
able to solve several query problems on this tree. In this section, 
we present all query problems that will be used in later sections.  

\begin{lemma}  
\label{lemLCAetc} 
Let $T$ be a tree with $n$ nodes that is rooted at an arbitrary node. 
We can preprocess $T$ in $O(n)$ time, such that each of the following 
queries can be answered in $O(1)$ time: 
\begin{enumerate}
\item Given a node $u$ of $T$, return its level, denoted by $\level(u)$, 
      which is the number of edges on the path from $u$ to the root. 
\item Given two nodes $u$ and $v$ of $T$, report their lowest common 
      ancestor, denoted by $\lca(u,v)$. 
\item Given two nodes $u$ and $v$ of $T$, decide whether or not $u$ is 
      in the subtree rooted at $v$. 
\item Given two distinct nodes $u$ and $v$ of $T$, report the second 
      node on the path from $u$ to $v$. 
\item Given three nodes $u$, $v$, and $w$, decide whether or not $w$ is 
      on the path between $u$ and $v$. 
\end{enumerate}
\end{lemma}
\begin{proof} 
The first claim follows from the fact that by performing an $O(n)$--time 
pre-order traversal of $T$, we can compute $\level(u)$ for each node $u$. 
A proof of the second claim can be found in 
Harel and Tarjan~\cite{DBLP:journals/siamcomp/HarelT84}
and Bender and Farach-Colton~\cite{bf-lpr-00}. 
The third claim follows from the fact that $u$ is in the subtree rooted 
at $v$ if and only $\lca(u,v) = v$. 
A proof of the fourth claim can be found in 
Chazelle~\cite[Lemma~15]{ChazelleFreeTreeComputing}.
The fifth claim follows from the following observations. 
Assume that $u$ is in the subtree rooted at $v$. Then $w$ is on the path 
between $u$ and $v$ if and only if $\lca(u,w) = w$ and $w$ is in the 
subtree rooted at $v$. The case when $v$ is in the subtree rooted at 
$u$ is symmetric. Assume that $\lca(u,v) \not\in \{u,v\}$. 
Then $w$ is on the path between $u$ and $v$ if and only if $w$ is on 
the path between $u$ and $\lca(u,v)$ or $w$ is on the path between $v$ 
and $\lca(u,v)$. 
\end{proof} 

\subsection{Closest-Colour Queries in Trees} 
Let $T$ be a tree with $n$ nodes and let $\mathcal{C}$ be a set of 
\emph{colours}. For each colour $c$ in $\mathcal{C}$, we are given a 
path $P_c$ in $T$. Even though these paths may share nodes, each node 
of $T$ belongs to at most a constant number of paths. This implies that 
the total size of all paths $P_c$ is $O(n)$. We assume that each node 
$u$ of $T$ stores the set of all colors $c$ such that $u$ is on the 
path $P_c$. 

In a \emph{closest-colour query}, we are given two nodes $u$ and $v$ 
of $T$, and a colour $c$, such that $u$ is on the path $P_c$. The answer 
to the query is the node on $P_c$ that is closest to $v$. 
Refer to Figure~\ref{fig:colour} for an illustration. 

\begin{figure}[h]
  \centering
  \includegraphics[scale = 0.5]{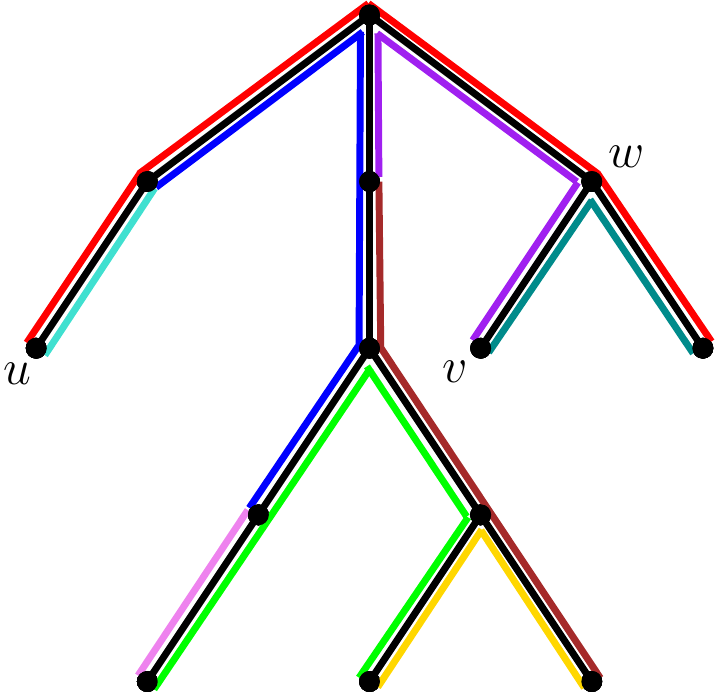}
  \caption{A tree $T$ and a collection of coloured paths. 
      For a query with nodes $u$ and $v$, and color ``red'', the 
      answer is the node $w$.}
  \label{fig:colour}
\end{figure}

\begin{lemma} 
\label{lem-last-label}
After an $O(n)$--time preprocessing, we can answer any closest-colour
query in $O(1)$ time. 
\end{lemma} 
\begin{proof} 
We take an arbitrary node of $T$ and make it the root. Then we preprocess 
$T$ such that each of the queries in Lemma~\ref{lemLCAetc} can be 
answered in $O(1)$ time.

For each colour $c$, let $c^1$ and $c^2$ be the end nodes of the path 
$P_c$, and let $c^h$ be the highest node on $P_c$ in the tree 
(i.e., the node on $P_c$ that is closest to the root). 
With each node of $P_c$, we store pointers to $c^1$, $c^2$, and $c^h$.

Since each node of $T$ is in a constant number of coloured paths, we 
can compute the pointers for all the coloured paths in $O(n)$ total 
time. 

The query algorithm does the following. Let $u$ and $v$ be two nodes 
of $T$, and let $c$ be a colour such that $u$ is on the $c$-coloured 
path $P_c$. 

If $u=v$ or $v$ is also on $P_c$, then we return the node $v$. From 
now on, assume that $u \neq v$ and $v$ is not on $P_c$. 
Below, we consider all possible cases, which are illustrated in 
Figure~\ref{fig:lem4}.

\begin{figure}[h]
  \centering
  \includegraphics[scale = 0.7]{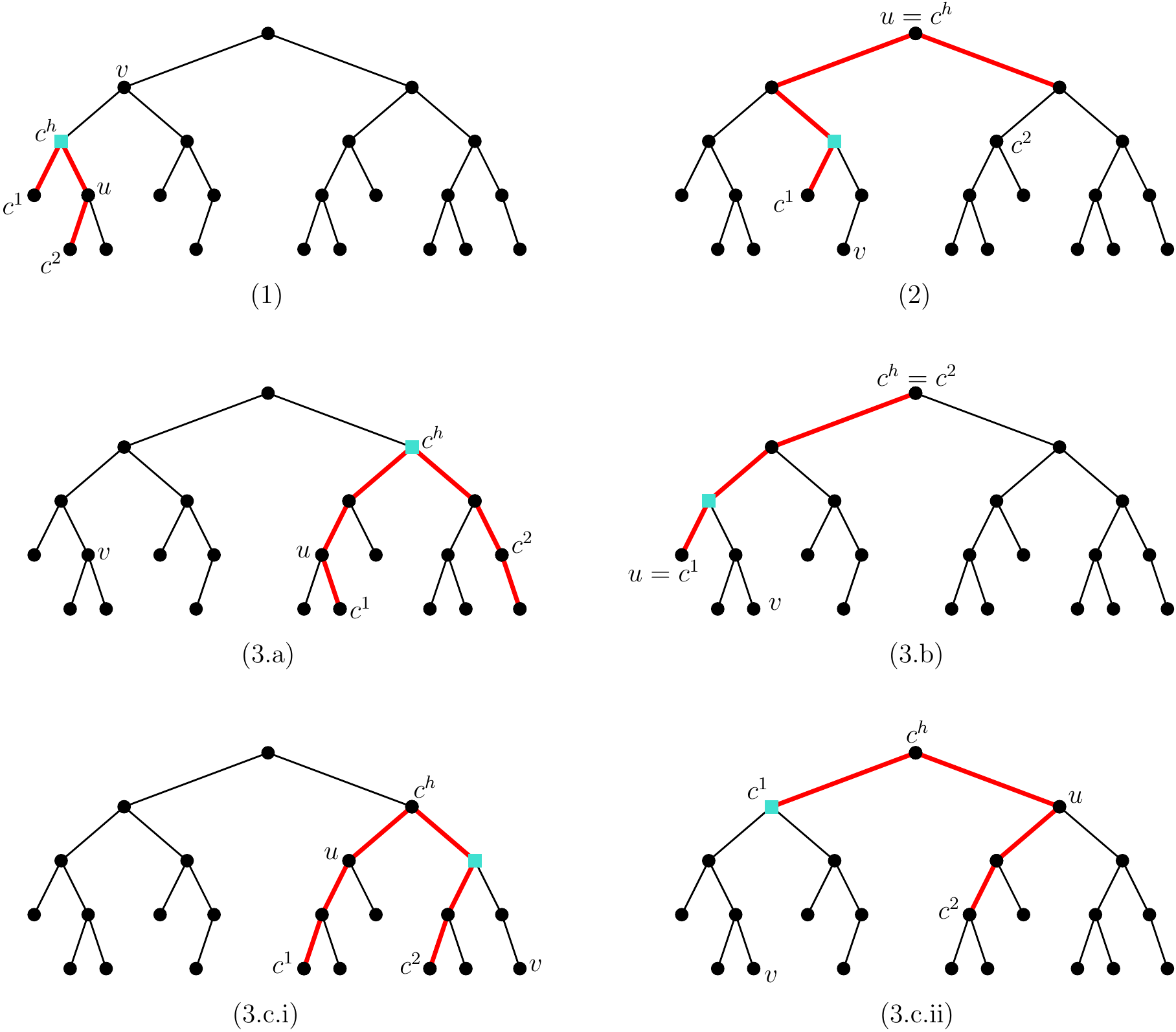}
  \caption{Illustrating all possible cases in the proof of 
           Lemma~\ref{lem-last-label}. The path $P_c$ is red and the 
           blue square indicates the node that is returned by the 
           closest-colour query.}
  \label{fig:lem4}
\end{figure}

\begin{enumerate}
\item If $\lca(u,v) = v$, then $u$ is in the subtree rooted at $v$. 
      In this case, we return $c^h$, the highest $c$-coloured node. 
\item Assume that $\lca(u,v) = u$. Then $v$ is in the subtree 
      rooted at $u$. The closest $c$-coloured node to $v$ is either
      $\lca(v,c^1)$ or $\lca(v,c^2)$. Since $v$ is lower than $u$ in the 
      tree, we know that the closest $c$-colored node to $v$ is at 
      $\level(u)$ or greater. If 
      $\level(\lca(v,c^1))>\level(\lca(v,c^2))$, 
      then $\lca(v,c^1)$ is lower in the tree and closer to $v$, so we 
      return $\lca(v,c^1)$ . Otherwise, $\lca(v,c^2)$ is lower in $T$
      or equal to both $\lca(v,c^1)$ and $u$, so we return  
      $\lca(v,c^2)$.
 \item Assume that $\lca(u,v) \neq u$ and $\lca(u,v) \neq v$. Then 
      $u$ and $v$ are in different subtrees of $\lca(u,v)$.      
      \begin{enumerate}
      \item If $\level(c^h) > \level(\lca(u,v))$, then we 
            return $c^h$. 
      \item If $\level(c^h) < \level(\lca(u,v))$, then we 
            return $\lca(u,v)$. 
      \item Assume that $\level(c^h) = \level(\lca(u,v))$. 
            Observe that exactly one end node of the $c$-coloured path 
            is in the subtree rooted at $u$. 
            \begin{enumerate}
            \item If $c^1$ is in the subtree rooted at 
                  $u$, then we return $\lca(v,c^2)$. 
            \item If $c^2$ is in the subtree rooted at 
                  $u$, then we return $\lca(v,c^1)$. 
            \end{enumerate}
      \end{enumerate}
\end{enumerate}

Using Lemma~\ref{lemLCAetc}, each of these case takes $O(1)$ time. 
Therefore, the entire query algorithm takes $O(1)$ time. 
\end{proof} 

\subsection{Path-Sum Queries in Trees} \label{secpathsum}
Let $(W,\oplus)$ be a semigroup. Thus, $W$ is a set and 
$\oplus : W \times W \rightarrow W$ is an associative binary operator. 
We assume that for any two elements $s$ and $s'$ in $W$, the value of 
$s \oplus s'$ can be computed in $O(1)$ time. 

Let $T$ be a tree with $n$ nodes in which each edge $e$ stores a value 
$s(e)$, which is an element of $W$. For any two distinct nodes $u$ and 
$v$ in $T$, we define their \emph{path-sum} $\ps(u,v)$ as follows: 
Let $e_1,e_2,\ldots,e_k$ be the edges on the path in $T$ between $u$ 
and $v$. Then we define $\ps(u,v) = \oplus_{i=1}^k s(e_i)$. 

Chazelle~\cite{ChazelleFreeTreeComputing} considers the problem of 
preprocessing the tree $T$, such that for any two distinct query 
nodes $u$ and $v$, the value of $\ps(u,v)$ can be reported. 
(See also Alon and Schieber~\cite{AlonOptimalPreprocessing}, 
Thorup~\cite{ThorupParallelShortcutting}, and 
Chan \emph{et al.}~\cite{timothy}.) Chazelle's result is stated in terms 
of the inverse Ackermann function; see Section~\ref{secOR}. 

\begin{lemma} \label{lemChazelle}  
Let $T$ be a tree with $n$ nodes in which each edge stores an element 
of the semigroup $(W,\oplus)$. 
For any integer $m \geq n$, we can preprocess $T$ in $O(m)$ time 
into a data structure of size $O(m)$, such that any path-sum query can 
be answered in $O(\alpha(m,n))$ time.   
\end{lemma}

\begin{remark} \label{remark2}
\emph{Assume that $(W,\oplus)$ is the semigroup, where $W$ is the set 
of all real numbers and the operator $\oplus$ takes the minimum of its 
arguments. In this case, we will refer to a query as a 
\emph{path-minimum query}. For this semigroup, the result of 
Lemma~\ref{lemChazelle} is optimal: Any data structure that can be 
constructed in $O(m)$ time has worst-case query time 
$\Omega(\alpha(m,n))$. To prove this, assume that we can answer any 
query in $o(\alpha(m,n))$ time. Then the on-line minimum spanning tree 
verification problem on a tree with $n$ vertices and $m \geq n$ 
queries can be solved in $o(m \cdot \alpha(m,n))$ time, 
by performing a path-maximum query for the endpoints of
each edge $e$ and checking that the weight of $e$ is larger
than the path-maximum. This
contradicts the lower bound for this problem proved by 
Pettie~\cite{PettieInverseAckermann}. 
}
\end{remark}

\section{Beer Distance Queries in Maximal Outerplanar Graphs}
\label{secmaxdist}

Let $G$ be a maximal outerplanar beer graph with $n$ vertices that 
satisfies the generalized triangle inequality. We will show how to 
preprocess $G$, such that for any two vertices $u$ and~$v$, the weight, 
$\dst_B(u,v)$, of a shortest beer path between $u$ and $v$ can be 
reported. Our approach will be to define a special semigroup 
$(W,\oplus)$, such that each element of $W$ ``contains'' certain 
distances and beer distances. With each edge of the dual $D(G)$, 
we will store one element of the set $W$. As we will see later, a beer 
distance query can then be reduced to a path-sum query in $D(G)$. 
Thus, by applying the results of Section~\ref{secpathsum}, we 
will obtain a proof of Theorem~\ref{thmdq}. 

We will need the first claim in the following lemma. The second claim 
will be used in Section~\ref{secJoyce}. 

\begin{lemma} 
\label{lem-crazy} 
Consider the beer graph $G$ as above. 
\begin{enumerate}
\item In $O(n)$ total time, we can compute $\dst_B(u,u)$ for each 
      vertex $u$ of $G$, and $\dst_B(u,v)$ for each edge $(u,v)$ in $G$. 
\item After an $O(n)$--time preprocessing of $G$, we can report, 
      \begin{enumerate}
      \item for any query edge $(u,v)$ of $G$, the shortest beer path 
            between $u$ and $v$ in $O(L)$ time, where $L$ is the number 
            of vertices on this path, 
      \item for any query vertex $u$ of $G$, the shortest beer path from 
            $u$ to itself in $O(L)$ time, where $L$ is the number of 
            vertices on this path. 
      \end{enumerate}
\end{enumerate}
\end{lemma} 
\begin{proof}
We choose an arbitrary face $R$ of $G$ and make it the root of $D(G)$. 
Let $(u,v)$ be any edge of $G$. This edge divides $G$ into two outerplanar 
subgraphs, both of which contain $(u,v)$ as an edge. Let $G_{uv}^R$ be 
the subgraph that contains the face $R$,
and let $G_{uv}^{\neg R}$ denote the other subgraph. Note that if 
$(u,v)$ is an external edge, then $G_{uv}^R = G$ and $G_{uv}^{\neg R}$ 
consists of the single edge $(u,v)$. 
By the generalized triangle inequality, the shortest beer path between 
$u$ and $v$ is completely in $G_{uv}^R$ or completely in 
$G_{uv}^{\neg R}$. The same is true for the shortest beer path from $u$ 
to itself. Thus, for each edge $(u,v)$ of $G$,
\[ \dst_B(u,v) = \min \left( \dst_B(u,v,G_{uv}^R) , 
                                 \dst_B(u,v,G_{uv}^{\neg R}) \right) , 
\]
\[ \dst_B(u,u) = \min \left( \dst_B(u,u,G_{uv}^R) , 
                                 \dst_B(u,u,G_{uv}^{\neg R}) \right) .  
\]

By performing a post-order traversal of $D(G)$, we can compute 
$\dst_B(u,v,G_{uv}^{\neg R})$ and $\dst_B(u,u,G_{uv}^{\neg R})$ for all 
edges $(u,v)$, in $O(n)$ total time. After these values have been 
computed, we perform a pre-order traversal of $D(G)$ and obtain 
$\dst_B(u,v,G_{uv}^R)$ and $\dst_B(u,u,G_{uv}^R)$, again for all edges 
$(u,v)$, in $O(n)$ total time. The details will be given in 
Section~\ref{app:lem-crazy}. 
\end{proof} 

In the rest of this section, we assume that all beer distances in the 
first claim of Lemma~\ref{lem-crazy} have been computed. 

For any two distinct internal faces $F$ and $F'$ of $G$, let $Q_{F,F'}$ 
be the union of the two sets  
\[ \{ (u,v,\dst(u,v),\dd) \mid u \mbox{ is a vertex of } F , 
          v \mbox{ is a vertex of } F'\}  
\]
and 
\[ \{ (u,v,\dst_B(u,v),\bsd) \mid u \mbox{ is a vertex of } F , 
          v \mbox{ is a vertex of } F'\} ,  
\]
where the ``bits'' $\dd$ and $\bsd$ indicate whether the tuple 
represents a distance or a beer distance. In words, $Q_{F,F'}$ is the 
set of all shortest path distances and all shortest beer distances 
between a vertex in $F$ and a vertex in $F'$. Since each internal face 
has three vertices, the set $Q_{F,F'}$ has exactly $18$ elements. 

\begin{observation} 
Let $u$ and $v$ be vertices of $G$, and let $F$ and $F'$ be internal 
faces that contain $u$ and $v$ as vertices, respectively. 
\begin{enumerate}
\item If $F=F'$, then we can determine both $\dst(u,v)$ and 
$\dst_B(u,v)$ in $O(1)$ time. 
\item If $F \neq F'$ and we are given the set $Q_{F,F'}$, then we can 
determine both $\dst(u,v)$ and $\dst_B(u,v)$ in $O(1)$ time. 
\end{enumerate}
\end{observation}
\begin{proof} 
First assume that $F = F'$. If $u=v$, then $\dst(u,v) = 0$ and 
$\dst_B(u,v)$ has been precomputed. If $u \neq v$, then $(u,v)$ is an 
edge of $G$ and, thus, $\dst(u,v) = \omega(u,v)$ and $\dst_B(u,v)$ has 
been precomputed. 

Assume that $F \neq F'$. If we know the set $Q_{F,F'}$, then we can 
find $\dst(u,v)$ and $\dst_B(u,v)$ in $O(1)$ time, because these two 
distances are in $Q_{F,F'}$. 
\end{proof} 

In the rest of this section, we will show that Lemma~\ref{lemChazelle} 
can be used to compute the set $Q_{F,F'}$ for any two distinct internal 
faces $F$ and $F'$. 

\begin{lemma}  \label{lembasecase} 
For any edge $(F,F')$ of $D(G)$, the set $Q_{F,F'}$ can be 
computed in $O(1)$ time. 
\end{lemma} 
\begin{proof} 
Let $u$ be a vertex of $F$ and let $v$ be a vertex of $F'$. Consider the 
subgraph $G[F,F']$ of $G$ that is induced by the four vertices of $F$ 
and $F'$; this subgraph has five edges. By the generalized triangle 
inequality, $\dst(u,v) = \dst(u,v,G[F,F'])$. Thus, $\dst(u,v)$ can be 
computed in $O(1)$ time. 

We now show how $\dst_B(u,v)$ can be computed in $O(1)$ time. 
If $u=v$ or $(u,v)$ is an edge of $G$, then $\dst_B(u,v)$ has been 
precomputed. Assume that $u \neq v$ and $(u,v)$ is not an edge of $G$. 
Let $w$ and $w'$ be the two vertices that are shared by $F$ and $F'$. 
Since any path in $G$ between $u$ and $v$ contains at least one of $w$ 
and $w'$, $\dst_B(u,v)$ is the minimum of 
\begin{enumerate}
\item $\dst_B(u,w) + \omega(w,v)$,
\item $\omega(u,w) + \dst_B(w,v)$, 
\item $\dst_B(u,w') + \omega(w',v)$, 
\item $\omega(u,w') + \dst_B(w',v)$. 
\end{enumerate}
Since $(u,w)$, $(w,v)$, $(u,w')$, and $(w',v)$ are edges of $G$, all 
terms in these four sums have been precomputed. Therefore, $\dst_B(u,v)$ 
can be computed in $O(1)$ time.

We have shown that each of the $18$ elements of $Q_{F,F'}$ can be 
computed in $O(1)$ time. Therefore, this entire set can be computed in 
$O(1)$ time.  
\end{proof} 

\begin{lemma}  \label{lemOplus} 
Let $F$, $F'$, and $F''$ be three pairwise distinct internal faces of
$G$, such that $F'$ is on the path in $D(G)$ between $F$ and $F''$.
If we are given the sets $Q_{F,F'}$ and $Q_{F',F''}$, then the set 
$Q_{F,F''}$ can be computed in $O(1)$ time. 
\end{lemma} 
\begin{proof}
Let $u$ be a vertex of $F$ and let $v$ be a vertex of $F''$. Since 
$G$ is an outerplanar graph, any path in $G$ between $u$ and $v$ 
must contain at least one vertex of $F'$. It follows that 
\[ \dst(u,v) = \min 
   \{ \dst(u,w) + \dst(w,v) \mid w \mbox{ is a vertex of } F' \} . 
\]
Thus, since $(u,w,\dst(u,w),\dd) \in Q_{F,F'}$ and 
$(w,v,\dst(w,v),\dd) \in Q_{F',F''}$, the value of $\dst(u,v)$ can be 
computed in $O(1)$ time. 

\begin{figure}[h]
  \centering
  \includegraphics[scale = 0.7]{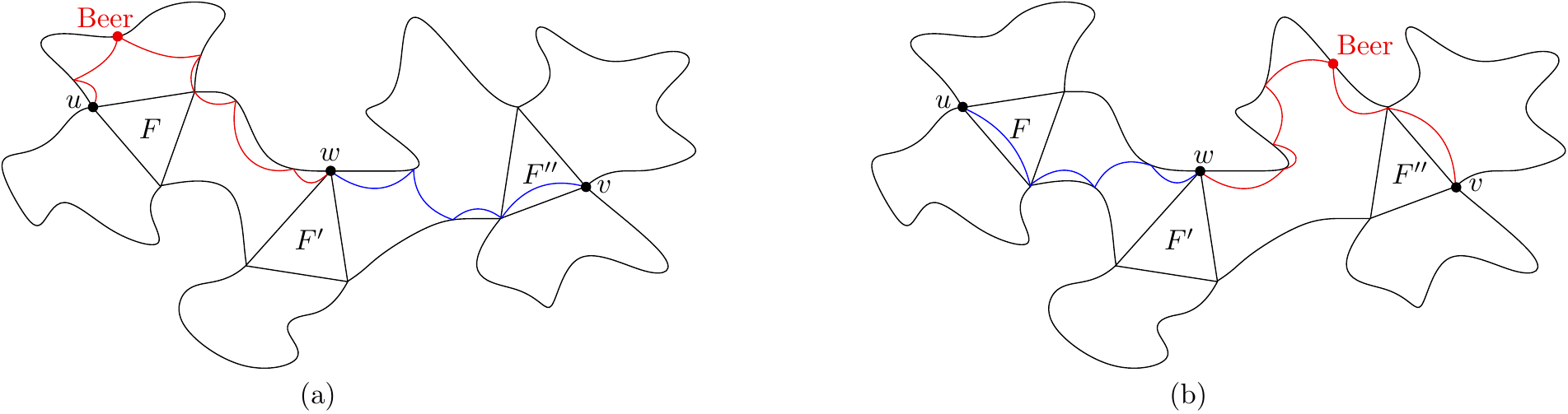}
  \caption{Any beer path from $u$ to $v$ contains at least one vertex 
      of $F'$. In (a), we consider the shortest beer path from $u$ to 
      $w$, followed by the shortest path from $w$ to $v$. 
      In (b), we consider the shortest path from $u$ to 
      $w$, followed by the shortest beer path from $w$ to $v$.}
  \label{fig:lemma10}
\end{figure}

By a similar argument, $\dst_B(u,v)$ is equal to (refer to 
Figure~\ref{fig:lemma10}) 
\[ \min \{ \min ( \dst_B(u,w) + \dst(w,v) , \dst(u,w) + \dst_B(w,v) ) : 
              \mbox{ $w$ is a vertex of $F'$}
        \} .  
\]
All values $\dst(u,w)$, $\dst(w,v)$, $\dst_B(u,w)$, and $\dst_B(w,v)$ 
are encoded in the sets $Q_{F,F'}$ and $Q_{F',F''}$. Therefore, we 
can compute $\dst_B(u,v)$ in $O(1)$ time. 

Thus, since each of the $18$ elements of $Q_{F,F''}$ can be computed 
in $O(1)$ time, the entire set can be computed in $O(1)$ time.  
\end{proof} 

We define 
\[ W = \{ Q_{F,F'} \mid \mbox{ $F$ and $F'$ are distinct internal faces 
          of $G$} \} \cup \{ \bot \} ,
\] 
where $\bot$ is a special symbol. We define the operator 
$\oplus : W \times W \rightarrow W$ in the following way.
\begin{enumerate}
\item If $F$ and $F'$ are distinct internal faces of $G$, then 
$Q_{F,F'} \oplus Q_{F,F'} = Q_{F,F'}$. 
\item If $F$, $F'$, and $F''$ are pairwise distinct internal faces of 
$G$ such that $F'$ is on the path in $D(G)$ between $F$ and $F''$, then 
$Q_{F,F'} \oplus Q_{F',F''} = Q_{F,F''}$. 
\item In all other cases, the operator $\oplus$ returns $\bot$.   
\end{enumerate}

It is not difficult to verify that $\oplus$ is associative, implying 
that $(W,\oplus)$ is a semigroup. 
By Lemma~\ref{lembasecase}, we can compute $Q_{F,F'}$ for all edges 
$(F,F')$ of $D(G)$, in $O(n)$ total time. 

Recall from Lemma~\ref{lemLCAetc} that, after an $O(n)$--time 
preprocessing, we can decide in $O(1)$ time, for any three internal 
faces $F$, $F'$, and $F''$ of $G$, whether $F'$ is on the path in 
$D(G)$ between $F$ and $F''$. Therefore, using Lemma~\ref{lemOplus}, 
the operator $\oplus$ takes $O(1)$ time to evaluate for any two 
elements of $W$.  

Finally, let $F$ and $F'$ be two distinct internal faces of $G$, and let 
$F = F_0 , F_1, F_2 , \ldots, F_k = F'$ be the path in $D(G)$ between 
$F$ and $F'$. Then $Q_{F,F'} = \oplus_{i=0}^{k-1} Q_{F_i,F_{i+1}}$. 
Thus, if we store with each edge of the tree $D(G)$, the corresponding 
element of the semigroup, then computing $Q_{F,F'}$ becomes a path-sum 
query as in Section~\ref{secpathsum}. 

To summarize, all conditions to apply Lemma~\ref{lemChazelle} are 
satisfied. As a result, we have proved Theorem~\ref{thmdq} for maximal 
outerplanar graphs that satisfy the generalized triangle inequality. 

\subsection{The Result in Theorem~\ref{thmdq} is Optimal}
In Section~\ref{secpathsum}, see also Remark~\ref{remark2}, we have seen
path-minimum queries in a tree, in which each edge $e$ stores a real
number $s(e)$. In such a query, we are given two distinct nodes $u$ and 
$v$, and have to return the smallest value $s(e)$ among all edges $e$ 
on the path between $u$ and $v$. Lemma~\ref{lemChazelle} gives a 
trade-off between the preprocessing and query times when answering such 
queries.

Let $D$ be an arbitrary data structure that answers beer distance
queries in any beer tree. Let $P(n)$, $S(n)$, and $Q(n)$ denote the
preprocessing time, space, and query time of $D$, respectively,
when the beer tree has $n$ nodes. We will show that $D$ can be used to
answer path-minimum queries.

Consider an arbitrary tree $T$ with $n$ nodes, such that each edge $e$
stores a real number $s(e)$. We may assume without loss of generality
that $0 < s(e) < 1$ for each edge $e$ of $T$.

By making an arbitrary node the root of $T$, the number of edges on the
path in $T$ between two nodes $u$ and $v$ is equal to
\[ \level(u) + \level(v) - 2 \cdot \level(\lca(u,v)) .
\]
Thus, by Lemma~\ref{lemLCAetc}, after an $O(n)$--time preprocessing,
we can compute the number of edges on this path in $O(1)$ time.

We create a beer tree $T'$ as follows. Initially, $T'$ is a copy of $T$.
For each edge $e=(u,v)$ of $T'$, we introduce a new node $x_e$ and
replace $e$ by two edges $(u,x_e)$ and $(v,x_e)$; we assign a weight
of $1$ to each of these two edges. In the current tree $T'$, none
of the nodes has a beer store. For every node $x_e$ in $T'$, we
introduce a new node $x'_e$, add the edge $(x_e,x'_e)$, assign a weight
of $s(e)$ to this edge, and make $x'_e$ a beer store. Finally, we
construct the data structure $D$ for the resulting beer tree $T'$.
Since $T'$ has $n+2(n-1) = 3n-2$ nodes, it takes $P(3n-2) + O(n)$ time
to construct $D$ from the input tree $T$. Moreover, the amount of
space used is $S(3n-2) + O(n)$.

Let $u$ and $v$ be two distinct nodes in the original tree $T$, let 
$\pi$ be the path in $T$ between $u$ and $v$, and let $\ell$ be the 
number of edges on $\pi$. The corresponding path $\pi'$ in $T'$ 
between $u$ and $v$ has weight $2 \ell$.

For any edge $e$ of $T$, let $\pi'_e$ be the beer path in $T'$ that
starts at $u$, goes to $x_e$, then goes to $x'_e$ and back to $x_e$, and
continues to $v$.

If $e$ is an edge of $\pi$, then the weight of $\pi'_e$ is equal to
$2 \ell + 2 \cdot s(e)$, which is less than $2 \ell + 2$. On the other 
hand, if $e$ is an edge of $T$ that is not on $\pi$, then the weight of
$\pi'_e$ is at least $2 \ell + 2 + 2 \cdot s(e)$, which is larger than 
$2 \ell + 2$. It follows that the shortest beer path in $T'$ between 
$u$ and $v$ visits the beer store $x'_e$, where $e$ is the edge on 
$\pi$ for which $s(e)$ is minimum.

Thus, by computing $\ell$ and querying $D$ for the beer distance in $T'$
between $u$ and $v$, we obtain the smallest value $s(e)$ among all edges
$e$ on the path in $T$ between $u$ and $v$. The query time is
$Q(3n-2) + O(1)$.

By combining this reduction with Remark~\ref{remark2}, it follows that
the result of Theorem~\ref{thmdq} is optimal.

\section{Reporting Shortest Beer Paths in Maximal Outerplanar Graphs}
\label{secJoyce} 

Let $G$ be a maximal outerplanar beer graph with $n$ vertices that 
satisfies the generalized triangle inequality. In this section, we show 
that, after an $O(n)$--time preprocessing, we can report, for any two 
query vertices $s$ and $t$, the shortest beer path $\spp_B(s,t)$ from 
$s$ to $t$, in $O(L)$ time, where $L$ is the number of vertices on 
this path. As before, $D(G)$ denotes the dual of $G$.  

\begin{observation} 
\label{obs:pathInD}
Let $v$ be a vertex of $G$. The faces of $G$ containing $v$ form 
a path of nodes in $D(G)$.
\end{observation} 

Define $P_v$ to be the path in $D(G)$ formed by the faces of $G$ 
containing the vertex $v$. Let $G[P_v]$ be the subgraph of $G$ induced 
by the faces of $G$ containing $v$. Note that $G[P_v]$ has a fan shape. 
Let $\cw(v)$ denote the clockwise neighbor of $v$ in $G[P_v]$ and 
let $\ccw(v)$ denote the counterclockwise neighbor of $v$ in $G[P_v]$. 
We will refer to the clockwise path from $\cw(v)$ to $\ccw(v)$ in 
$G[P_v]$ as the $v$-\emph{chain} and denote it by $\rho_v$. 
(Refer to Figure~\ref{fig:v-chain}.) 

\begin{figure}[h]
  \centering
  \includegraphics[scale = 0.7]{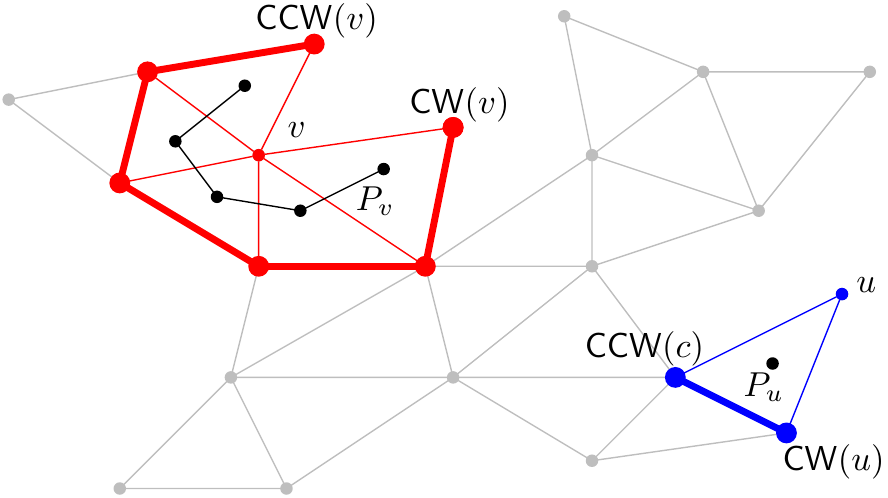}
  \caption{A maximal outerplanar graph $G$. The subgraphs $G[P_v]$ and 
    $G[P_u]$ are shown in red and blue, respectively. Both the $v$-chain 
    $\rho_v$ and the $u$-chain $\rho_u$ are shown in bold. Both paths 
    $P_v$ and $P_u$ are shown in black. Observe that $P_u$ is a single 
    node.}
  \label{fig:v-chain}
\end{figure}

\begin{lemma}
\label{shortestPathInAPath}
After an $O(n)$--time preprocessing, we can answer the following 
queries, for any three query vertices $v$, $u$, and $w$, such that both 
$u$ and $w$ are on the $v$-chain $\rho_v$: 
\begin{enumerate}
\item Report the weight $\dst(u,w,\rho_v)$ of the path from $u$ to 
$w$ along $\rho_v$ in $O(1)$ time. 
\item Report the path $\spp(u,w,\rho_v)$ from $u$ to $w$ along 
$\rho_v$ in $O(L)$ time, where $L$ is the number of vertices on 
this path. 
\end{enumerate}
\end{lemma}
\begin{proof}
For any vertex $v$ and any vertex $u$ on $\rho_v$, we store the weight
of the path from $u$ to $\cw(v)$ along $\rho_v$. Observe that 
\[ \dst(u,w,\rho_v) = 
     | \dst(u,\cw(v),\rho_v) - \dst(w,\cw(v),\rho_v) | .  
\]
 
Any exterior edge in $G$ is in exactly one chain and any 
interior edge in $G$ is in exactly two chains. Thus, the sum of the number
of edges on each chain is proportional to the number of edges of $G$, which is 
$O(n)$. 
\end{proof}

\begin{lemma}
\label{lem:distInG(v)}
After an $O(n)$--time preprocessing, we can answer the following query 
in $O(1)$ time: Given three query vertices $v$, $u$, and $w$, such that 
both $u$ and $w$ are vertices of $G[P_v]$, report $\dst(u,w)$, i.e., 
the distance between $u$ and $w$ in $G$. 
\end{lemma}
\begin{proof}
We get the following cases; the correctness follows from the 
generalized triangle inequality:
\begin{enumerate}
\item If $u=w$ then $\dst(u,w)=0$.
\item If $u=v$ then $(u,w)$ is an edge and we return $\omega(u,w)$. 
Similarly if $w=v$, we return $\omega(u,w)$.
\item Otherwise $u$ and $w$ are both on $\rho_v$ and we return 
$\min(\dst(u,w,\rho_v), \omega(u,v)+\omega(v,w))$. 
\end{enumerate}
\end{proof}

\begin{lemma}
\label{lem:SPInG(Pv)}
After an $O(n)$--time preprocessing, we can report, for any three 
vertices $v$, $u$, and $w$, such that both $u$ and $w$ are vertices of 
$G[P_v]$, $\spp(u,w)$ in $O(L)$ time, where $L$ is the number of vertices 
on the path. 
\end{lemma}
\begin{proof}
Using Lemma~\ref{lem:distInG(v)}, we can determine in $O(1)$ if the 
shortest path from $u$ to $w$ goes through $v$ or follows the $v$-chain
$\rho_v$. (Refer to Figure~\ref{fig:twoCases}). If it goes 
through $v$, then $\spp(u,w)=(u,v,w)$. Otherwise, 
$\spp(u,w)$ takes the path along $\rho_v$ and by 
Lemma~\ref{shortestPathInAPath}, we can find this path in $O(L)$ time. 
\end{proof}

\begin{figure}[h]
  \centering
  \includegraphics[scale = 0.5]{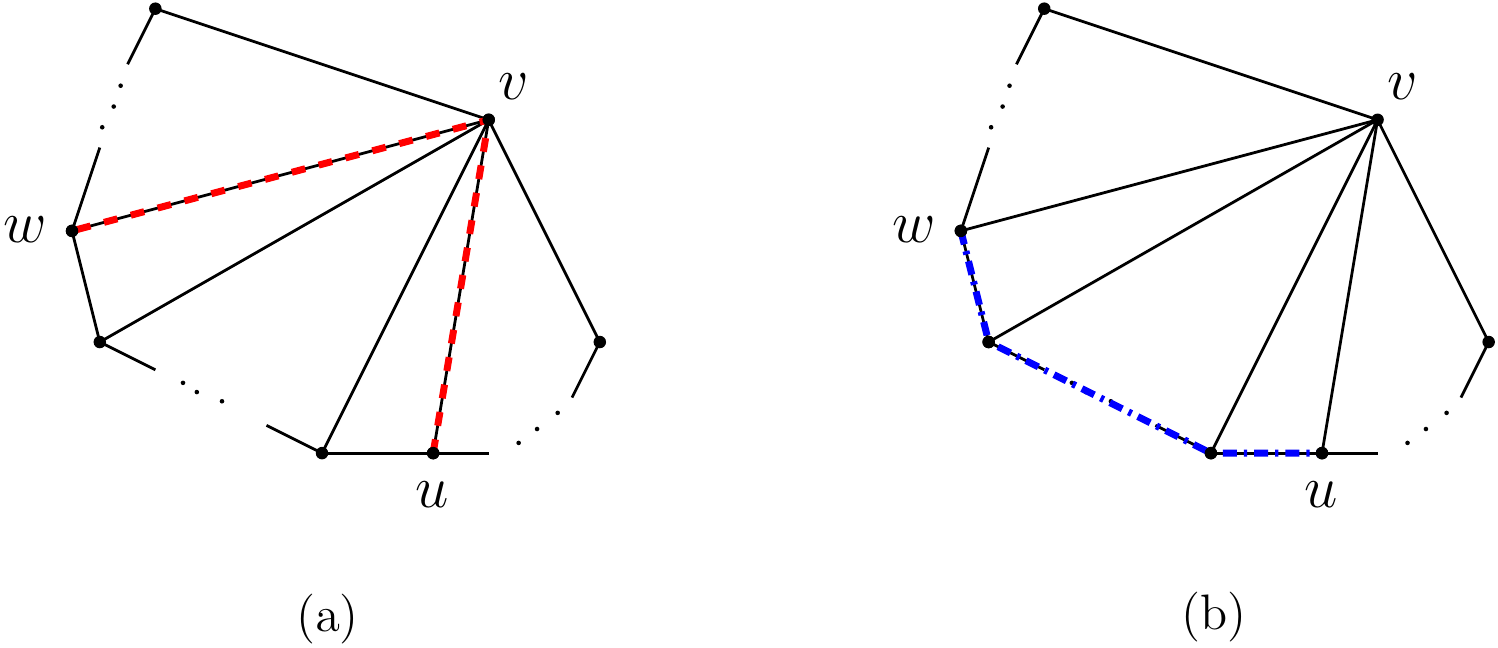}
  \caption{Two possible cases for the shortest path between $u$ and $w$: (a) it goes through vertex $v$ (shown in dashed red), or (b) it goes through the vertices of the $v$-chain between $u$ and $w$ (shown in dashed blue).}
  \label{fig:twoCases}
\end{figure}

\begin{lemma}
\label{lem:beerDistInG(v)}
After an $O(n)$--time preprocessing, we can report, for any three 
vertices $v$, $u$ and $w$, such that both $u$ and $w$ are vertices of 
$G[P_v]$, the beer distance $\dst_B(u,w)$ in $O(1)$ time. The 
corresponding shortest beer path $\spp_B(u,w)$ can be reported in 
$O(L)$ time, where $L$ is the number of vertices on the path.
\end{lemma}
\begin{proof}
Recall from Lemma~\ref{lem-crazy} that we can compute $\dst_B(u,v)$ for 
every edge $(u,v)$ in $G$, and $\dst_B(v,v)$ for every vertex $v$ in $G$, 
in $O(n)$ time. 

Let $\rho_v=(\cw(v)=u_1, u_2, \dots, u_N=\ccw(v))$. Let $A_v[~]$ be 
an array of size $N-1$. For $i=1,\ldots,N-1$, we set 
$A_v[i]=\dst_B(u_i,u_{i+1})- \omega(u_i,u_{i+1})$. Recall that by 
the generalized triangle inequality, 
$\omega(u_i,u_{i+1})=\dst(u_i,u_{i+1})$. Therefore, $A[i]$ holds the 
difference between the weights of the shortest path from $u_i$ to 
$u_{i+1}$ and the shortest beer path from $u_i$ to $u_{i+1}$. After 
preprocessing the array $A_v[~]$ in $O(N)$ time, we can conduct range 
minimum queries in $O(1)$ time. 
(Bender and Farach-Colton~\cite{bf-lpr-00} show that these queries are 
equivalent to $\lca$-queries in the Cartesian tree of the array.) 
Thus, for each $v$-chain of $N$ nodes, we spend $O(N)$ time processing 
the $v$-chain. Since every edge is in at most two chains, processing all 
$v$-chains takes $O(n)$ time and space.

Given two vertices $u$ and $w$ of $G[P_v]$, we determine the beer 
distance $\dst_B(u,w)$ as follows:
\begin{enumerate}
\item If $u=w$  then $\dst_B(u,w)$ has already been computed by 
Lemma~\ref{lem-crazy}.
\item If $u=v$ or $w=v$, then there is an edge from $v$ to the other 
vertex. Thus, $\dst_B(u,w)$ has already been computed by 
Lemma~\ref{lem-crazy}.
\item Otherwise, $u$, $w$ and $v$ are three distinct vertices. Assume 
without loss of generality that $w$ is clockwise from $u$ on the 
$v$-chain. We take the minimum of the following two cases:
	\begin{enumerate}
	\item The shortest beer path from $u$ to $w$ that goes through 
            $v$. Since a beer store must be visited before or after 
            $v$, this beer path has a weight of 
            $\min(\dst_B(u,v)+\omega(v,w),\omega(u,v)+\dst_B(v,w))$.
	 \item The shortest beer path through the vertices of the 
            $v$-chain. Note that this beer path will visit each vertex 
            on the $v$-chain between $u$ and $w$, but may go off the 
            $v$-chain to visit a beer store. On $\spp_B(u,w)$, there is 
            one pair of vertices, $u_i$ and $u_{i+1}$, such that a 
            beer path is taken between $u_i$ and $u_{i+1}$, and $u_i$ 
            and $u_{i+1}$ are adjacent on the $v$-chain; refer to 
            Figure~\ref{fig:v-pathWithBeer}. 
	    The shortest path is taken between all other pairs of 
            adjacent vertices on the $v$-chain. From
	    Lemma~\ref{shortestPathInAPath}, we can compute 
            $\dst(u,w,\rho_v)$ in $O(1)$ time. The shortest beer path 
            through the vertices of the $v$-chain has a weight of 
	    $\dst(u,w,\rho_v)+A_v[i]$, where $A_v[i]$ is the additional 
            distance needed to visit a beer store between $u_i$ and 
            $u_{i+1}$. Let $u$ be the $j^{th}$ vertex on $\rho_v$
	    and let $w$ be the $k^{th}$ vertex in $\rho_v$. Then 
            $A_v[i]$ is the minimum value in $A_v[j, \dots, k-1]$. 
            We can determine $A_v[i]$ in constant time using a range 
            minimum query.
	\end{enumerate} 
\end{enumerate}
Note that in case 1 and case 2, $\spp_B(u,w)$ can be constructed in $O(L)$ time by 
Lemma \ref{lem-crazy}. For case 3 (a) let $p=(u,v,w)$ and for case 3 (b) let $p=\spp(u,w,\rho_v)$. 
Let $u_i$, $u_{i+1}$ be the pair of adjacent vertices 
on $p$ between which a beer path was taken. Using Lemma \ref{lem-crazy} we can find 
$\spp_B(u_i, u_{i+1})$ in $O(L)$ time. We obtain $\spp_B(u,w)$ by replacing  
the edge $(u_i, u_{i+1})$ in $p$ with  $\spp_B(u_i, u_{i+1})$.
\end{proof}

\begin{figure}[h]
  \centering
  \includegraphics[scale = 0.7]{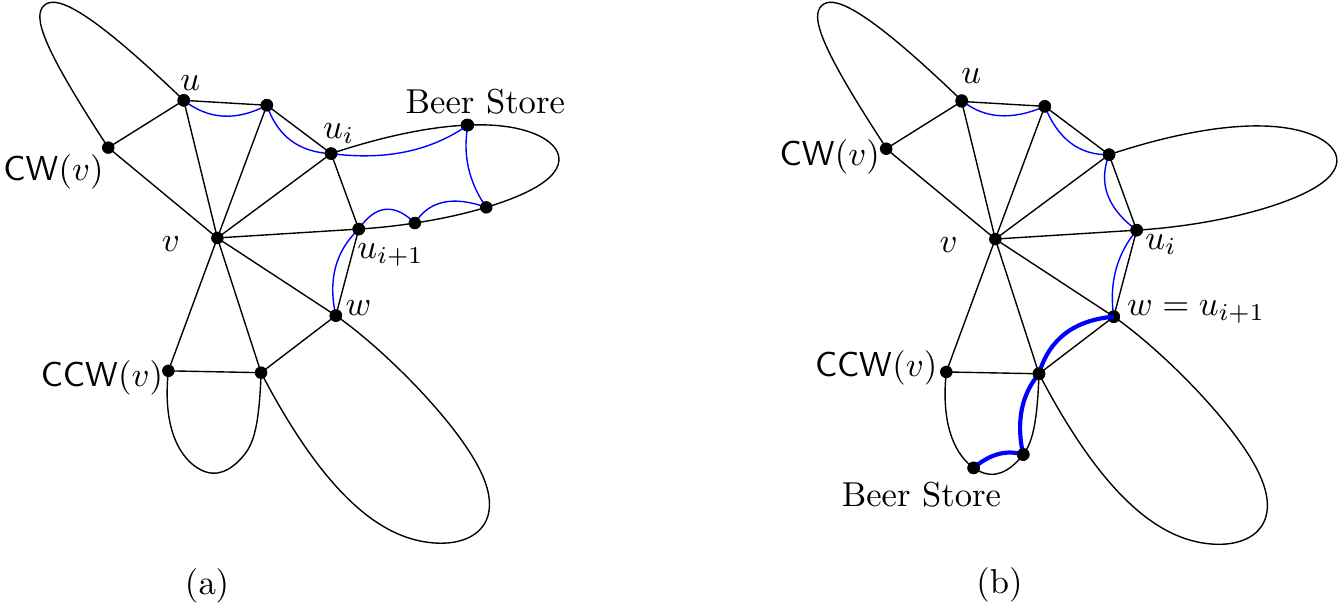}
  \caption{Both figures show a shortest beer path from $u$ to $w$ through the vertices on the $v$-chain. Thicker edges on the blue beer path are edges that are traversed twice; once in each direction.}
  \label{fig:v-pathWithBeer}
\end{figure}

\subsection{Answering Shortest Beer Path Queries}
Recall that, for any vertex $v$ of $G$, $P_v$ denotes the path in $D(G)$ 
formed by the faces of $G$ containing $v$. Moreover, $G[P_v]$ denotes 
the subgraph of $G$ induced by these faces. 

Consider two query vertices $s$ and $t$ of $G$. Our goal is to compute 
the shortest beer path $\spp_B(s,t)$. 

Let $F_s$ and $F_t$ be arbitrary faces containing $s$ and $t$, 
respectively. 
If $t$ is in $G[P_s]$ then, by Lemma \ref{lem:beerDistInG(v)}, 
we can construct $\spp_B(s,t)$ in $O(L)$ time. For the remainder of 
this section, we assume that $t$ is not in $G[P_s]$. To find 
$\spp_B(s,t)$, we start by constructing a directed acyclic graph (DAG), 
$H$. In this DAG, vertices will be arranged in columns of constant 
size, and all edges go from left to right between vertices in adjacent 
columns. In $H$, each column will contain one vertex that is on 
$\spp_B(s,t)$. First we will construct $H$ and then we will show how 
we can use $H$ to construct $\spp_B(s,t)$. 
The entire construction is illustrated in Figure~\ref{fig:dag}. 

\begin{observation}
\label{oneVertexOnPath}
Any interior edge $(a,b)$ of $G$ splits $G$ into two subgraphs such 
that if $s$ is in one subgraph and $t$ is in the other, then any path 
in $G$ from $s$ to $t$ must visit at least one of $a$ and $b$. 
\end{observation}

Let $P$ be the unique path between $F_s$ and $F_t$ in $D(G)$. Consider 
moving along $P$ from $F_s$ to $F_t$. Let $F_1$ be the node on 
$P_s$ that is closest to $F_t$, and let $F'_1$ be the successor 
of $F_1$ on $P$. Note that, by Lemmas~\ref{lemLCAetc} 
and~\ref{lem-last-label}, we can find $F_1$ and $F'_1$ in $O(1)$ 
time.\footnote{To apply Lemma~\ref{lem-last-label}, we consider each 
vertex of $G$ to be a colour. For each vertex $v$ of $G$, the 
$v$-coloured path in the tree $D(G)$ is the path $P_v$. The face $F_1$ 
is the answer to the closest-colour query with nodes $F_s$ and $F_t$
and colour $s$.} 
Let  $e_1=(a_1,b_1)$ be the edge in $G$ shared by the faces $F_1$ 
and $F'_1$. Since $\spp_B(s,t)$ must visit both of these faces, by 
Observation \ref{oneVertexOnPath}, at least one of $a_1$ or $b_1$ is 
on the shortest beer path. 

We place $s$ in the first column of $H$ and $a_1$ and $b_1$ in the 
second column of $H$. We then add two directed edges from $s$ to $a_1$, 
one with weight $\dst(s,a_1)$ and the other with weight $\dst_B(s,a_1)$. 
Similarly, we add two directed edges from $s$ to $b_1$ with weights 
$\dst(s,b_1)$ and $\dst_B(s,b_1)$. 

\begin{figure}[h]
  \centering
  \includegraphics[scale = 0.7]{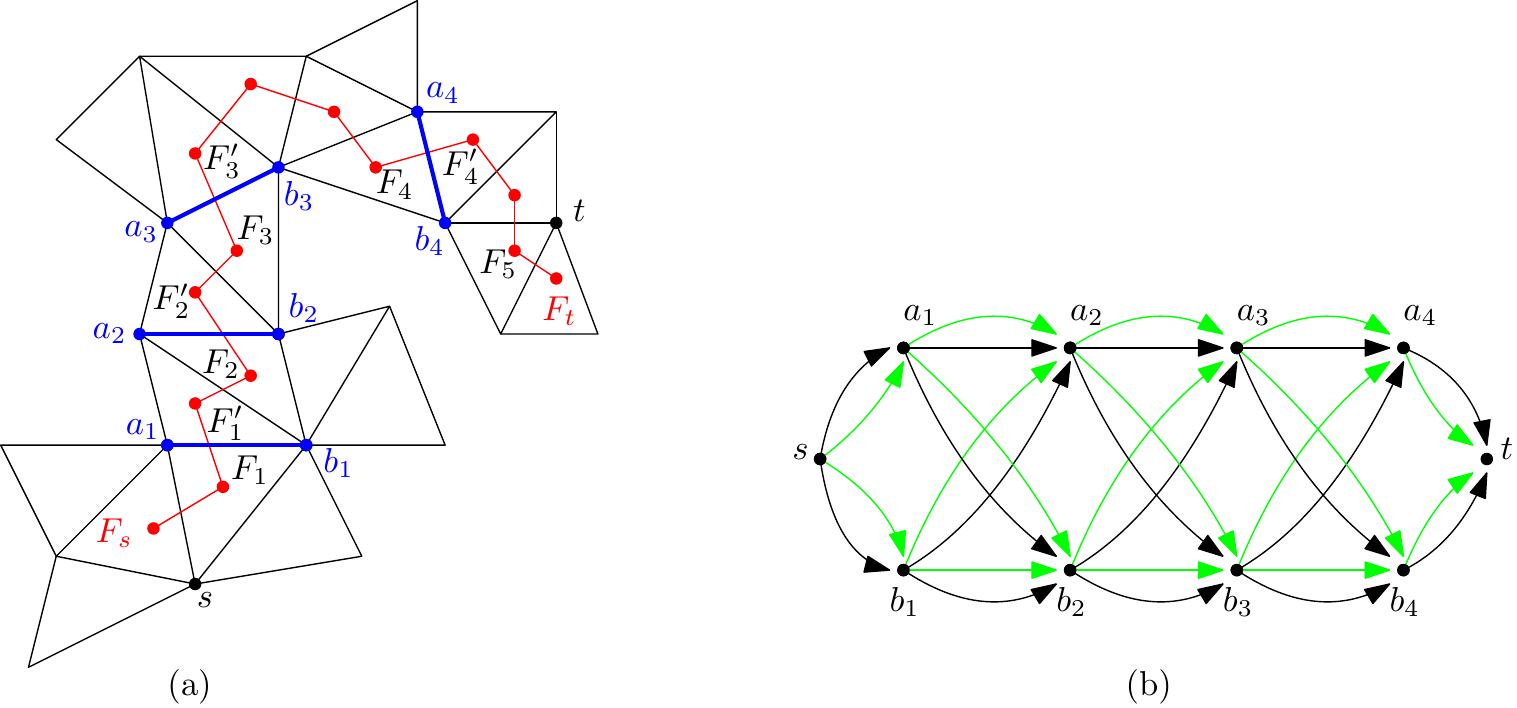}
  \caption{An outerplanar graph $G$ (a) and the DAG $H$ constructed 
   for the shortest beer path query from $s$ to $t$ (b). The path $P$ 
   from $F_s$ to $F_t$ is shown in red. Each edge $e_i=(a_i,b_i)$ such 
   that $e_i$ is shared by $F_i$ and $F'_i$ is shown in blue. The green 
   edges of $H$ represent the beer edges.}
  \label{fig:dag}
\end{figure}

When $i \geq 2$ we construct the $(i+1)^{th}$ column of $H$ in the 
following way. Let $e_{i-1}=(a_{i-1},b_{i-1})$ be the edge shared by 
the faces $F_{i-1}$ and $F'_{i-1}$. The $i^{th}$ column of $H$ contains 
the vertices $a_{i-1}$ and $b_{i-1}$. Note that $F'_{i-1}$ is in both 
$P_{b_{i-1}}$ and $P_{a_{i-1}}$. Using Lemma~\ref{lem-last-label}, we 
find the node $F_{i}^b$ on $P_{b_{i-1}}$ that is closest to $F_t$. 
If the vertex $a_{i-1}$ is not in $F_{i}^b$, then we let 
$F_{i}=F_{i}^b$. Otherwise, we let $F_{i}$ be the node on $P_{a_{i-1}}$ 
that is closest to $F_t$. 

If $t$ is not a vertex of $F_{i}$, then let $F'_i$ be the node that 
follows $F_i$ on $P$; we find $F'_i$ using Lemma~\ref{lemLCAetc}. 
Let $e_ i=(a_ i,b_ i)$ be the edge of $G$ shared by the faces $F_i$ 
and $F'_i$. 
In the $(i+1)^{th}$ column, we place $a_i$ and $b_i$. 
For each $u\in \{a_{i-1},b_{i-1}\}$ and each $v \in \{a_i,b_i\}$ we add two 
directed edges $(u,v)$ to the DAG, one with weight $\dst(u,v)$ and the 
other with weight $\dst_B(u,v)$. If $F_i$ is in $P_{a_{i-1}}$, all these 
vertices are in $G[P_{a_{i-1}}]$; otherwise, $F_i$ is in $P_{b_{i-1}}$, 
and all these vertices are in $G[P_{b_{i-1}}]$. Thus, by 
Lemmas~\ref{lem:distInG(v)} and~\ref{lem:beerDistInG(v)}, we can find 
the distances and beer distances to 
assign to these edges in constant time. 

If $t$ is in $F_{i}$, then in the $(i+1)^{th}$ column we only place 
the vertex $t$. In this case, for each  $u\in \{a_{i-1},b_{i-1}\}$, 
we add two directed edges $(u,t)$ to the DAG with weights $\dst(u,t)$ 
and $\dst_B(u,t)$. At this point we are done constructing $H$.

We define a \emph{beer edge} to be an edge of $H$ that was assigned 
a weight of a beer path during the construction of $H$. We find the 
beer distance from $s$ to $t$ in $G$ using the following dynamic 
programming approach in $H$. 

Let $M$ denote the number of columns in $H$. For $i=3,\dots ,M$ and 
for all $u$ in the $i^{th}$ column of $H$, compute
\[ \dst_B(s,u) = \min 
        \begin{cases}
           \dst_B(s,a_{i-2})+\dst(a_{i-2},u), \\
           \dst(s,a_{i-2})+\dst_B(a_{i-2},u), \\
           \dst_B(s,b_{i-2})+\dst(b_{i-2},u), \\
           \dst(s,b_{i-2})+\dst_B(b_{i-2},u)
        \end{cases} 
\]
and
\[ \dst(s,u) = \min 
        \begin{cases}
           \dst(s,a_{i-2})+\dst(a_{i-2},u), \\
           \dst(s,b_{i-2})+\dst(b_{i-2},u). \\
        \end{cases}
\] 

The vertices $a_{i-2}$ and $b_{i-2}$ occur in the $(i-1)^{th}$ column. 
Thus, $\dst_B(s,a_{i-2})$, $\dst_B(s,b_{i-2})$, $\dst(s,a_{i-2})$, and  
$\dst(s,b_{i-2})$ will be computed before computing the values for the
$i^{th}$ column. We get $\dst(a_{i-2},u)$, $\dst_B(a_{i-2},u)$, 
$\dst(b_{i-2},u)$ and $\dst_B(b_{i-2},u)$ from the weights of the 
DAG-edges between the $(i-1)^{th}$ and $i^{th}$ columns of $H$. 

By keeping track of which expression produced $\dst_B(s,u)$ and 
$\dst(s,u)$, we can backwards reconstruct the shortest beer path 
in the DAG. Knowing the shortest beer path in the DAG enables us to
construct the corresponding beer path in $G$ as follows.
\begin{enumerate}
\item Define $P_{st}$ to be an empty path.
\item For each edge $(w,v)$ of the shortest beer path in the DAG.
	\begin{enumerate}
	\item  If $(w,v)$ was a beer edge, let $P_{wv}=\spp_B(w,v)$, 
	which can be constructed in time proportional to its number 
        of vertices via Lemma \ref{lem:beerDistInG(v)}.
	\item  Otherwise, let $P_{wv}=\spp(w,v)$ which can be 
	constructed in time proportional to its number of vertices as seen in
	Lemma \ref{lem:SPInG(Pv)}.
	\end{enumerate} 
Let $P_{st}=P_{st} \cup P_{wv}$.
\item Return $P_{st}$, which is equal to $\spp_B(w,v)$.
\end{enumerate}

Let $L$ denote the number of vertices on $\spp_B(s,t)$. 
In order for the above query algorithm to take $O(L)$ time, the size of 
the DAG must be $O(L)$. The following three lemmas will show this to be 
true. 

\begin{lemma}\label{lem-a-or-b}
For $2\leq i<M-1$, $F_{i}$ contains either $a_{i-1}$ or $b_{i-1}$, 
but not both.   
\end{lemma}
\begin{proof}
Recall that we defined $F_i^b$ to be the last node on $P$ that is also on 
$P_{b_{i-1}}$. We similarly define $F_i^a$ to be the last node on $P$ 
that is also on $P_{a_{i-1}}$. From the way we choose $F_i$, $F_i$ is 
either $F_i^b$ or $F_i^a$. We only choose $F_{i}=F_i^b$ after having 
checked that $a_{i-1}$ is not in $F_i^b$; thus in this case we can be 
sure that $F_{i}$  only contains $b_{i-1}$.

Assume for the purpose of contradiction that we choose $F_i=F_i^a$ and 
$b_{i-1}$ is also in $F_i$. Let the third vertex of $F_i$ be $c$. Let 
the face on $P$ immediately following $F_i$ be $F_i'$. The edge shared 
by $F_i$ and $F'_i$ is either $(b_{i-1},c)$ or $(a_{i-1},c)$. If 
$(b_{i-1},c)$ is the shared edge, then $F_i'$ is a face closer to 
$F_t$ that contains $b_{i-1}$ and not $a_{i-1}$, so we would have 
chosen $F_i=F_i^b$, which is a contradiction. Otherwise, $(a_{i-1},c)$ 
is the edge shared by $F_i$ and $F_i'$, which implies that there is 
a face containing $a_{i-1}$ closer to $F_t$ in $P$ than $F_i^a$, 
which contradicts the definition of $F_i^a$.
\end{proof} 

\begin{lemma} \label{lem-unique}
Every vertex of $G$ appears in at most one column of $H$.
\end{lemma}
\begin{proof} 
Since $(a_1,b_1)$ is an edge shared by both the last face of $P$ containing $s$ 
and the first face of $P$ that does not contain $s$ it is not possible for either of these 
vertices to be the vertex $s$. Thus, $s$ will only be 
represented by the vertex in the first column of $H$. By stopping
the construction of $H$ as soon as we add a vertex representing $t$, we ensure 
that $H$ only contains one vertex corresponding to the vertex $t$ in $G$.

For $2\leq i \leq M-2$, consider the vertex $a_{i-1}$ in $G$ represented by a 
vertex in the $i^{th}$ column of $H$. If $F_i=F_i^a$ then by definition of 
$F_i^a$, $F'_i$ does not contain $a_{i-1}$. Since $(a_i,b_i)$ is an edge of $F'_i$,
$a_i \neq a_{i-1}$ and $b_i \neq a_{i-1}$. Because the face $F'_i$ is closer 
to $F_t$ than $F_i^a$, $a_{i-1}$ is not a vertex on any of the faces on the 
path from $F'_i$ to $F_t$. Thus, subsequent columns of $H$ will not contain vertices  
representing the vertex $a_{i-1}$ in $G$.

If $ F_i=F_i^b$ then by Lemma \ref{lem-a-or-b}, $a_{i-1}$ is not in $F_i$ and since 
$(a_i,b_i)$ is an edge of $F_i$, $a_i \neq a_{i-1}$ and $b_i \neq a_{i-1}$.
Because $F_i$ is a face on $P$ closer to $F_t$ than $F_{i-1}$ (a face that 
contains $a_{i-1}$) it follows from Observation  \ref{obs:pathInD}  that none of 
the faces on $P$ from $F_{i-1}$ to $F_t$ will have the vertex $a_{i-1}$ on their face and, 
thus, $a_{i-1}$ will not be represented by vertices in subsequent columns of $H$. 

By switching the roles of $a_{i-1}$ with $b_{i-1}$ in the above reasoning we can 
see that this also holds for $b_{i-1}$.
\end{proof} 

\begin{lemma}
The number of vertices and edges of $H$ is $O(L)$.
\end{lemma}
 \begin{proof}
By Observation~\ref{oneVertexOnPath} and Lemma~\ref{lem-unique}, 
the number of columns of $H$ is at most $L$. Since each column has 
at most two vertices, each of which having at most four outgoing edges,  
the total number of vertices and edges of $H$ is $O(L)$. 
\end{proof}

Observe that the total preprocessing time is $O(n)$. For two query vertices 
$s$ and $t$, the DAG, $H$, can be constructed in $O(L)$ time. 
Finally, the dynamic programming algorithm on $H$ takes $O(L)$ time. 
Thus, we have proved Theorem~\ref{thm2} for maximal outerplanar graphs
that satisfy the generalized triangle inequality. 

\section{Proof of Lemma~\ref{lem-crazy}} 
\label{app:lem-crazy} 

Let $G$ be a maximal outerplanar beer graph with $n$ vertices that 
satisfies the generalized triangle inequality. We will first show how to 
compute $\dst_B(u,u)$ for each vertex $u$ of $G$, and $\dst_B(u,v)$ for 
each edge $(u,v)$ of $G$. Consider again the dual $D(G)$ of $G$. We choose 
an arbitrary face of $G$ and make it the root of $D(G)$. 

Let $(u,v)$ be any edge of $G$. This edge divides $G$ into two outerplanar 
subgraphs, both of which contain $(u,v)$ as an edge. Let $G_{uv}^R$ be 
the subgraph that contains the face represented by the root of $D(G)$,
and let $G_{uv}^{\neg R}$ denote the other subgraph. Note that if 
$(u,v)$ is an external edge, then $G_{uv}^R = G$ and $G_{uv}^{\neg R}$ 
consists of the single edge $(u,v)$. 

By the generalized triangle inequality, the shortest beer path between 
$u$ and $v$ is completely in $G_{uv}^R$ or completely in 
$G_{uv}^{\neg R}$. The same is true for the shortest beer path from $u$ 
to itself. This implies: 

\begin{observation} 
For each edge $(u,v)$ of $G$,
\begin{enumerate} 
\item $\dst_B(u,v) = \min \left( \dst_B(u,v,G_{uv}^R) , 
                                 \dst_B(u,v,G_{uv}^{\neg R}) \right)$,  
\item $\dst_B(u,u) = \min \left( \dst_B(u,u,G_{uv}^R) , 
                                 \dst_B(u,u,G_{uv}^{\neg R}) \right)$,  
\item $\dst_B(v,v) = \min \left( \dst_B(v,l,G_{uv}^R) , 
                                 \dst_B(v,v,G_{uv}^{\neg R}) \right)$,  
\end{enumerate} 
\end{observation} 

Thus, it suffices to first compute $\dst_B(u,v,G_{uv}^{\neg R})$,  
$\dst_B(u,u,G_{uv}^{\neg R})$, and $\dst_B(v,v,G_{uv}^{\neg R})$ 
for all edges $(u,v)$, and then compute $\dst_B(u,v,G_{uv}^R)$,
$\dst_B(u,u,G_{uv}^R)$, and $\dst_B(v,v,G_{uv}^R)$, again for all edges 
$(u,v)$.

\subsection{Recurrences for $\dst_B(u,v,G_{uv}^{\neg R})$, 
$\dst_B(u,u,G_{uv}^{\neg R})$, and $\dst_B(v,v,G_{uv}^{\neg R})$} 

Let $(u,v)$ be an edge of $G$. Item 1.\ below presents the base cases,
whereas item 2.\ gives the recurrences. 

\begin{figure}[h]
  \centering
  \includegraphics[scale = 0.5]{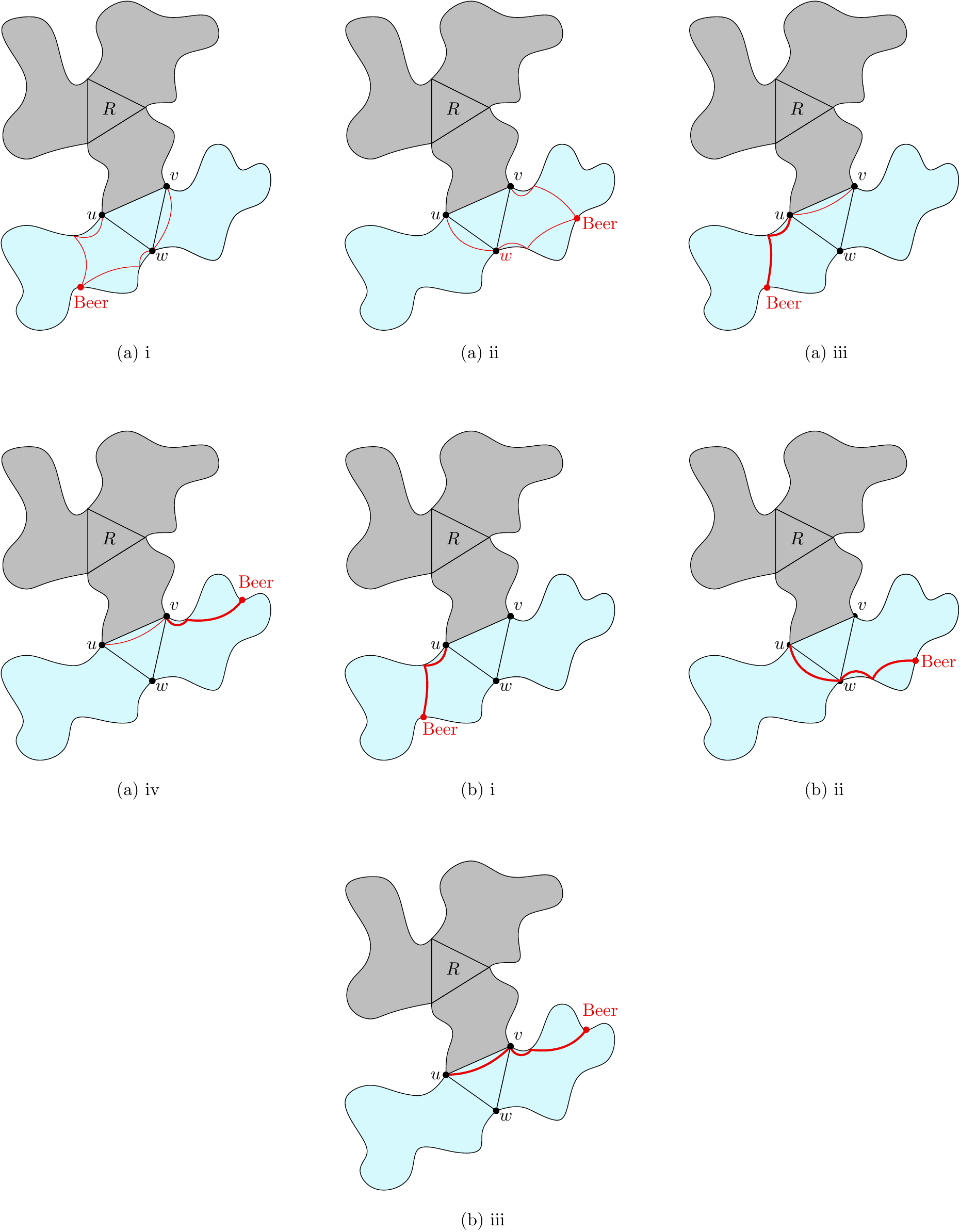}
  \caption{Illustrating the post-order traversal for all cases in 
           item 2.}
  \label{fig:post}
\end{figure}

\begin{enumerate}
\item Assume that $(u,v)$ is an external edge of $G$. 
      \begin{enumerate}
      \item If both $u$ and $v$ are beer stores, then 
            $\dst_B(u,v,G_{uv}^{\neg R}) = \omega(u,v)$, 
            $\dst_B(u,u,G_{uv}^{\neg R}) = 0$, and  
            $\dst_B(v,v,G_{uv}^{\neg R}) = 0$. 
      \item If exactly one of $u$ and $v$, say $u$, is a beer store, then 
            $\dst_B(u,v,G_{uv}^{\neg R}) = \omega(u,v)$, 
            $\dst_B(u,u,G_{uv}^{\neg R}) = 0$, and  
            $\dst_B(v,v,G_{uv}^{\neg R}) = 2 \cdot \omega(u,v)$. 
      \item If neither $u$ nor $v$ is a beer store, then 
            $\dst_B(u,v,G_{uv}^{\neg R}) = \infty$, 
            $\dst_B(u,u,G_{uv}^{\neg R}) = \infty$, and 
            $\dst_B(v,v,G_{uv}^{\neg R}) = \infty$. 
      \end{enumerate}
\item Assume that $(u,v)$ is an internal edge of $G$. Let $w$ be the 
      third vertex of the face of $G_{uv}^{\neg R}$ that contains $(u,v)$ 
      as an edge. All possible cases are illustrated in 
      Figure~\ref{fig:post}.
      \begin{enumerate}
      \item The value of $\dst_B(u,v,G_{uv}^{\neg R})$ is the minimum of 
            \begin{enumerate}
            \item $\dst_B(u,w,G_{uw}^{\neg R}) + \omega(w,v)$, 
            \item $\omega(u,w) + \dst_B(w,v,G_{vw}^{\neg R})$, 
            \item $\dst_B(u,u,G_{uw}^{\neg R}) + \omega(u,v)$, 
            \item $\omega(u,v) + \dst_B(v,v,G_{vw}^{\neg R})$. 
            \end{enumerate} 
      \item The value of $\dst_B(u,u,G_{uv}^{\neg R})$ is the minimum of 
            \begin{enumerate}
            \item $\dst_B(u,u,G_{uw}^{\neg R})$, 
            \item $2 \cdot \omega(u,w) + \dst_B(w,w,G_{vw}^{\neg R})$, 
            \item $2 \cdot \omega(u,v) + \dst_B(v,v,G_{vw}^{\neg R})$. 
            \end{enumerate} 
            The value of $\dst_B(v,v,G_{uv}^{\neg R})$ is obtained by 
            swapping $u$ and $v$ in i., ii., and iii. 
      \end{enumerate}
\end{enumerate} 

These recurrences express $\dst_B(u,v,G_{uv}^{\neg R})$, 
$\dst_B(u,u,G_{uv}^{\neg R})$, and $\dst_B(v,v,G_{uv}^{\neg R})$ in 
terms of values that are ``further down'' in the tree $D(G)$. Therefore, 
by performing a post-order traversal of $D(G)$, we obtain all these 
values, for all edges $(u,v)$ of $G$, in $O(n)$ total time.

\subsection{Recurrences for $\dst_B(u,v,G_{uv}^R)$, $\dst_B(u,u,G_{uv}^R)$, 
and $\dst_B(v,v,G_{uv}^R)$} 

Let $(u,v)$ be an edge of $G$. Item 1.\ below presents the base cases,
whereas item 2.\ gives the recurrences. 

\begin{enumerate}
\item Assume that $(u,v)$ is an edge of the face representing the root 
      of $D(G)$. Let $w$ be the third vertex of this face. 
      \begin{enumerate}
      \item The value of $\dst_B(u,v,G_{uv}^R)$ is the minimum of 
            \begin{enumerate}
            \item $\dst_B(u,u,G_{uw}^{\neg R}) + \omega(u,v)$, 
            \item $\dst_B(u,w,G_{uw}^{\neg R}) + \omega(w,v)$, 
            \item $\omega(u,v) + \dst_B(v,v,G_{vw}^{\neg R})$,
            \item $\omega(u,w) + \dst_B(w,v,G_{vw}^{\neg R})$.  
            \end{enumerate}
      \item The value of $\dst_B(u,u,G_{uv}^R)$ is the minimum 
            of\footnote{We do not have to consider 
            $W := \omega(u,v) + \dst_B(v,w,G_{vw}^{\neg R}) + 
            \omega(w,u)$, because the sum of the values in ii.\ and iii.\ 
            is at most $2W$. Therefore, the smaller of the values in ii.\ 
            and iii.\ is at most $W$.}  
            \begin{enumerate}
            \item $\dst_B(u,u,G_{uw}^{\neg R})$,
            \item $2 \cdot \omega(u,w) + \dst_B(w,w,G_{vw}^{\neg R})$,
            \item $2 \cdot \omega(u,v) + \dst_B(v,v,G_{vw}^{\neg R})$.
            \end{enumerate}
            The value of $\dst_B(v,v,G_{uv}^{\neg R})$ is obtained by 
            swapping $u$ and $v$ in i., ii., and iii. 
      \end{enumerate}
\item Assume that $(u,v)$ is not an edge of the face represented by the 
      root of $D(G)$. Let $w$ be the third vertex of the face of 
      $G_{uv}^R$ 
      that contains $(u,v)$ as an edge. We may assume without loss of 
      generality that $(v,w)$ is an edge of the face represented by the 
      parent of the face representing $(u,v,w)$. 
      All possible cases are illustrated in Figure~\ref{fig:pre}.

\begin{figure}[h]
  \centering
  \includegraphics[scale = 0.5]{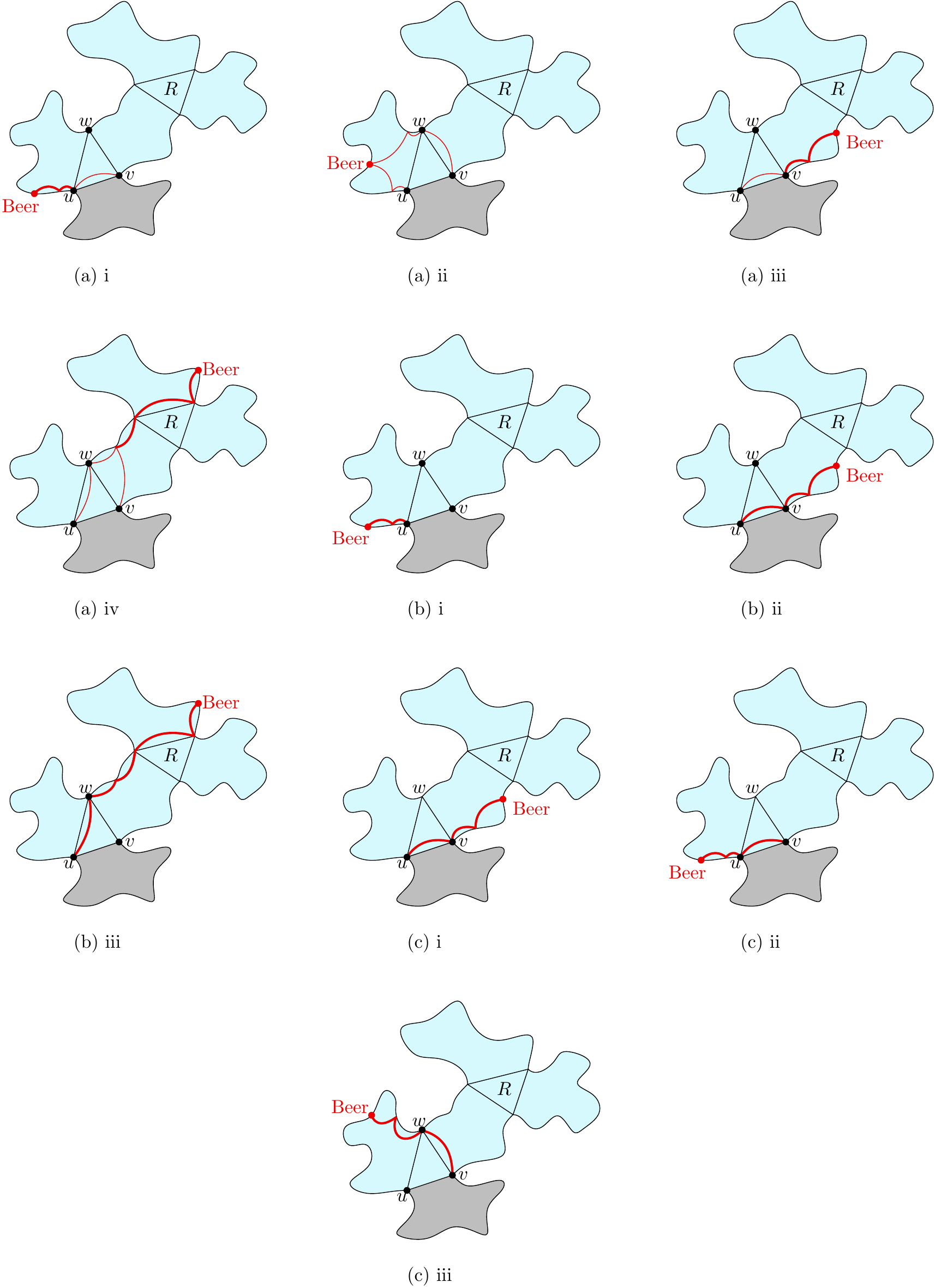}
  \caption{Illustrating the pre-order traversal for all cases in 
           item 2.}
  \label{fig:pre}
\end{figure}
      
      \begin{enumerate}
      \item The value of $\dst_B(u,v,G_{uv}^R)$ is the minimum of 
            \begin{enumerate}
            \item $\dst_B(u,u,G_{uw}^{\neg R}) + \omega(u,v)$,
            \item $\dst_B(u,w,G_{uw}^{\neg R}) + \omega(w,v)$, 
            \item $\omega(u,v) + \dst_B(v,v,G_{vw}^R)$, 
            \item $\omega(u,w) + \dst_B(w,v,G_{vw}^R)$. 
            \end{enumerate}
      \item The value of $\dst_B(u,u,G_{uv}^R)$ is the minimum of 
            \begin{enumerate}
            \item $\dst_B(u,u,G_{uw}^{\neg R})$,
            \item $2 \cdot \omega(u,v) + \dst_B(v,v,G_{vw}^R)$,
            \item $2 \cdot \omega(u,w) + \dst_B(w,w,G_{vw}^R)$.
            \end{enumerate}
      \item The value of $\dst_B(v,v,G_{uv}^R)$ is the minimum of 
            \begin{enumerate}
            \item $\dst_B(v,v,G_{vw}^R)$,
            \item $2 \cdot \omega(v,u) + \dst_B(u,u,G_{uw}^R)$,
            \item $2 \cdot \omega(v,w) + \dst_B(w,w,G_{uw}^R)$.
            \end{enumerate}
      \end{enumerate}
\end{enumerate} 

These recurrences express $\dst_B(u,v,G_{uv}^R)$, $\dst_B(u,u,G_{uv}^{R})$, 
and $\dst_B(v,v,G_{uv}^R)$ in terms of values that are ``higher up''
in the tree $D(G)$ and values that involve graphs with a superscript 
``$\neg R$''. These latter values have been computed already. 
Thus, by performing a pre-order traversal of $D(G)$, we obtain all values
$\dst_B(u,v,G_{uv}^R)$, $\dst_B(u,u,G_{uv}^R)$, and 
$\dst_B(v,v,G_{uv}^R)$, for all edges $(u,v)$ of $G$, in $O(n)$ total 
time. This completes the proof of the first claim in 
Lemma~\ref{lem-crazy}.

Consider $\dst(u,v)$ where $u=v$ or $(u,v)$ is an edge of $G$. If $u$ or $v$ is a beer store, then store $nil$ with $\dst(u,v)$. 

The values $\dst(u,v,G_{uv}^R)$ and $\dst(u,v,G_{uv}^{\neg R})$ are computed as the minimum of a set of path weights from $u$ to $v$ through a vertex $x$ such that $x$ is adjacent to both $u$ and $v$ or $x$ is equal to one of these vertices and adjacent to the other. Either the subpath from $u$ to $x$ is a beer path or the subpath from $x$ to $v$ is a beer path. Whenever we take the minimum of a set of path weights in the above computation, we store with that distance the vertex $x$ and a bit to indicate which subpath is the beer path.  When $u=v$ we can arbitrarily choose which subpath is the beer path. After taking the minimum of $\dst(u,v,G_{uv}^R)$ and $\dst(u,v,G_{uv}^{\neg R})$, we are left with a vertex, $x$, on the shortest beer path from $u$ to $v$ and the bit indicating which subpath is a beer path.

We recursively compute $\spp_B(u,v)$ where either $(u,v)$ is an edge of $G$ or $u=v$ as follows.
\begin{enumerate}
\item If $nil$ is stored with $\dst_B(u,v)$ and $u=v$, $\spp_B(u,v) = (u)$.
\item If $nil$ is stored with $\dst_B(u,v)$ and $(u,v)$ is an edge then $\spp_B(u,v) = (u,v)$.
\item If a vertex $x$ is stored with $(u,v)$ and the subpath from $u$ to $x$ is a beer path then recursively compute $\dst_B(u,x)$. $\spp_B(u,v)=\spp_B(u,x)\cup (v)$.
\item Otherwise $x$ is stored with $(u,v)$ and the subpath from $x$ to $v$ is a beer path. Recursively compute $\dst_B(x,v)$. $\spp_B(u,v) = (u) \cup \spp_B(x,v)$.
\end{enumerate}

Note that a constant amount of work is done at each level of the recurrence excluding the time spent in recursive calls. In each recursive call, except potentially the last call, we get one new vertex on the shortest beer path. Thus, constructing the whole path requires a total of $O(L)$ time.

\section{Extension to Arbitrary Outerplanar Graphs}
\label{app:extension} 


\subsection{Maximal Outerplanar}
Let $G$ be an outerplanar beer graph with $n$ vertices. 
Assume that the outer face of $G$ is not a Hamiltonian cycle. 
We traverse $G$ along the outer face in a clockwise manner, and mark 
each vertex when we encounter it for the first time. Each time we visit 
a marked vertex $v$, we take note of $v$'s current counterclockwise 
neighbor, $\ccw(v)$. Then we continue from $v$ to the next clockwise vertex on the 
outer face and add an edge from this vertex to $\ccw(v)$. We continue this 
process until we have returned to the vertex we started from and 
all vertices have been marked.

At this moment, the outer face is a Hamiltonian cycle. 
For every interior face that is not a triangle, we pick a vertex $u$ 
on that face and add edges connecting $u$ with all vertices of the 
face that are not already adjacent to $u$. 

The resulting graph is a maximal outerplanar graph. Each edge that has 
been added is given a weight of infinity. Observe that each shortest 
(beer) path in the resulting graph corresponds to a shortest (beer) 
path in the original graph, and vice versa. 

\subsection{Generalized Triangle Inequality}
Let $G$ be a maximal outerplanar graph with an edge weight function 
$\omega$. In order to convert $G$ to a graph that satisfies the 
generalized triangle inequality, we need to compute $\dst(u,v)$ for 
every edge $(u,v)$ in $G$, and we need to be able to construct 
$\spp(u,v)$ for each edge $(u,v)$. 

Let $D(G)$ be the dual of $G$ rooted at an arbitrary interior face of 
$G$. For each edge $(u,v)$, we initialize $\delta(u,v) = \omega(u,v)$.  
(At the end, $\delta(u,v)$ will be equal to $\dst(u,v)$.) 
For each edge $(u,v)$, we also maintain a parent vertex, $p(u,v)$, 
initialized to $nil$. 

We first conduct a post-order traversal of $D(G)$, processing each 
associated face in $G$. Then we conduct a pre-order traversal of $D(G)$,
again processing each associated face in $G$. Let $F$ be a face of $G$. 
We process $F$ as follows. Let $(u,v)$ be the edge of $F$ that is shared 
with the predecessor face $F'$ in the traversal, and let $w$ be the 
third vertex of $F'$. If $\delta(u,v) > \delta(u,w) + \delta(w,v)$, 
we set $\delta(u,v) = \delta(u,w) + \delta(w,v)$ and $p(u,v)=w$. 

After these traversals, $\delta(u,v) = \dst(u,v)$ for every edge 
$(u,v)$ in $G$. If $p(u,v) = nil$, then $\spp(u,v) = (u,v)$; otherwise, 
$\spp(u,v)$ is the concatenation of $\spp(u,p(u,v))$ and 
$\spp(p(u,v),v)$, both of which can be computed recursively.

\section{Single Source Shortest Beer Path}
\label{secSSSBP} 
In this section we will describe how to compute the single source shortest beer path from a source vertex $s$ on a maximal outerplanar graph $G$ that satisfies the generalized triangle inequality. In order to do this we first precompute (i) $\dst_B(u,v)$ for every edge $(u,v)$ and $\dst_B(u,u)$ for every vertex $u$ and (ii) $\dst(s,v)$ for every vertex $v$. By Lemma \ref{lem-crazy}, we can compute (i) in $O(n)$ time and in \cite{MaheshwariAndZeh}, Maheshwari and Zeh present a single source shortest path algorithm for undirected outerplanar graphs which gives us (ii) in $O(n)$ time. 

Let $D(G)$ be the dual of $G$ and let $F_s$ be an arbitrary interior face of $G$ containing $s$. Root $D(G)$ at the node $F_s$ and then conduct a pre-order traversal of $D(G)$. Let $F$ be the current node of $D(G)$ being processed during
this traversal. 

\begin{enumerate}
\item If $F=F_s$, let $u$ and $v$ be the vertices of $F_s$ that are not $s$. Since $(s,u)$ and $(s,v)$ are both edges of $G$, $\dst_B(s,s)$, $\dst_B(s,u)$ and $\dst_B(s,v)$ were precomputed in (i). 
\item If $F \neq F_s$, let $a$ and $b$ be the vertices of $F$ shared 
with the face $F'$, where $F'$ is the parent of $F$ in $D(G)$. This implies that by this step we have already computed $\dst_B(s,a)$ and $\dst_B(s,b)$. Let $c$ be the third vertex of $F$. The value of $\dst_B(s,c)$ is the minimum of:
\begin{enumerate}
\item $\dst(s,a)+\dst_B(a,c)$,
\item $\dst_B(s,a)+\omega(a,c)$,
\item $\dst(s,b)+\dst_B(b,c)$,
\item $\dst_B(s,b)+\omega(b,c)$.
\end{enumerate}
Since $(a,c)$ and $(b,c)$ are edges of $G$ we precomputed $\dst_B(a,c)$ and $\dst_B(b,c)$ in (i). Lastly, we computed $\dst(s,a)$ and $\dst(s,b)$ in (ii) so each of the values listed above can be computed in constant time.
\end{enumerate}

The correctness of this algorithm follows from 
Observation~\ref{oneVertexOnPath} and the generalized triangle 
inequality. Since we do a constant amount of work at each face in 
the traversal of $D(G)$ and the number of interior faces of a 
maximal outerplanar graph is $n-2$, this algorithm takes 
$O(n)$ time.

Let $L$ the the number of vertices on the shortest beer path from 
$s$ to $v$. If $\dst_B(s,v)$ was found in step~1, then by 
Lemma~\ref{lem-crazy}, $\spp_B(s,v)$ can be constructed in $O(L)$ 
time. If this is not the case, then $v=c$ in some iteration of step~2. 
At this step we store a vertex $p(v)$ such that $p(v)=a$ if (a) 
or (b) was the minimum of step~2 and $p(v)=b$ otherwise. We also store a bit to indicate if the subpath from $s$ to $p(v)$ is the shortest path (as in cases (a) and (c)) or the shortest beer path (as in cases (b) and (d)). If the subpath from $s$ to $p(v)$ is the shortest path, we use the method described in \cite{MaheshwariAndZeh} to find $\spp(s,p(v))$ and use Lemma \ref{lem-crazy} to find $\spp_B(p(v),v)$ and then concatenate $\spp(s,p(v))$ and  $\spp_B(p(v),v)$ to get $\spp_B(s,v)$. Both $\spp(s,p(v))$ and $\spp_B(p(v),v)$ are found in time proportional to the number of vertices on their paths, so this takes $O(L)$ time. If the subpath from $s$ to $p(v)$ is the shortest beer path, then we recursively find $\spp_B(s,p(v))$ and concatenate it with the edge $(p(v),v)$ (which is $\spp(p(v),v)$ by the generalized triangle inequality).  Each iteration of the recursive step takes time proportional to the number of new vertices of the path found in that step. Thus, we find  $\spp_B(s,v)$ in a total of $O(L)$ time.

\bibliographystyle{plain}
\bibliography{FinalVersion}

\end{document}